\documentclass[11pt]{amsart}

\usepackage{appendix} 

\usepackage{tabularx, environ}

\usepackage{amsmath, amsthm, amssymb, amsfonts} 

\allowdisplaybreaks 

\usepackage{eurosym} 

\usepackage{stmaryrd} 
\usepackage{mathabx} 
\usepackage{bm} 

\usepackage[hidelinks]{hyperref} 
\usepackage{color} 
\usepackage[normalem]{ulem} 
\usepackage{comment} 

\usepackage{graphicx} 
\usepackage{wrapfig} 
\usepackage{caption} 
\usepackage{subcaption} 
\usepackage{multirow} 
\usepackage{array} 
\usepackage[section]{placeins} 

\newtheorem{thm}{Theorem}[section]
\newtheorem{cor}[thm]{Corollary}
\newtheorem{lem}[thm]{Lemma}
\newtheorem{prop}[thm]{Proposition}

\theoremstyle{definition}
\newtheorem{defn}[thm]{Definition}

\theoremstyle{remark}
\newtheorem{rem}[thm]{Remark}

\newcommand{\R}{\mathbb{R}}
\newcommand{\C}{\mathbb{C}}
\newcommand{\N}{\mathbb{N}}

\newcommand{\ind}{\bm{1}}

\newcommand{\ud}{\mathrm{d}}

\newcommand{\im}{\ensuremath{\mathsf{i}}}

\renewcommand{\Re}{\mathrm{Re}}

\newcommand{\be}{\begin{equation}}
\newcommand{\ee}{\end{equation}}

\numberwithin{equation}{section}

\newcommand{\cD}{\mathcal{D}}
\newcommand{\cE}{\mathcal{E}}
\newcommand{\cF}{\mathcal{F}}

\newcommand{\cK}{\mathcal{K}}

\newcommand{\cN}{\mathcal{N}}	

\newcommand{\cP}{\mathcal{P}}

\newcommand{\cT}{\mathcal{T}}
\newcommand{\cU}{\mathcal{U}}
\newcommand{\cV}{\mathcal{V}}
\newcommand{\cW}{\mathcal{W}}

\newcommand{\FF}{\mathbb{F}}
\newcommand{\EE}{\mathbb{E}}

\newcommand{\QQ}{\mathbb{Q}}

\newcommand{\tildeN}{\widetilde{N}}

\newcommand{\uCBITCL}{\mathrm{CBITCL}}

\newcommand{\uUSD}{\mathrm{USD}}
\newcommand{\uEUR}{\mathrm{EUR}}
\newcommand{\uJPY}{\mathrm{JPY}}

\newcommand{\horizon}{\cT} 

\newcommand{\dbra}[1]{[\kern-0.15em[ #1 ]\kern-0.15em]}
\newcommand{\dbraco}[1]{[\kern-0.15em[ #1 [\kern-0.15em[}
\newcommand{\dbraoc}[1]{]\kern-0.15em] #1 ]\kern-0.15em]}
\newcommand{\dbraoo}[1]{]\kern-0.15em] #1 [\kern-0.15em[}

\title{CBI-time-changed L\'evy processes for multi-currency modeling}

\author[C. Fontana]{Claudio Fontana}
\address{Department of Mathematics ``Tullio Levi Civita'', University of Padova (Italy)}
\email{fontana@math.unipd.it}

\author[A. Gnoatto]{Alessandro Gnoatto}
\address{Department of Economics, University of Verona (Italy)}
\email{alessandro.gnoatto@univr.it}

\author[G. Szulda]{Guillaume Szulda}
\address{Department of Mathematics ``Tullio Levi Civita'', University of Padova (Italy), and Laboratoire de Probabilit\'es, Statistique et Mod\'elisation, Universit\'e de Paris Cit\'e (France)}
\email{szulda.guillaume@gmail.com}

\subjclass[2010]{60G51, 60J85, 91G20, 91G30, 91G60}
\thanks{\textit{JEL Classification}: C02, C60, G13, G15.}
\keywords{FX market; multi-currency market; branching process; self-exciting process; time-change; stochastic volatility; deep calibration; affine process.}
\thanks{{\em Acknowledgements: }C.F. is grateful to the Europlace Institute of Finance for financial support to this work. G.S. acknowledges hospitality and financial support from the University of Verona, where part of this work has been conducted. This work is part of the project  ``Term structure dynamics in interest rate and energy markets: modeling and numerics'' (BIRD190200/19) funded by the University of Padova. We are thankful to two anonymous Reviewers for useful comments that helped to improve the paper.}

\date{\today}

\begin{document}

\begin{abstract}
We develop a stochastic volatility framework for modeling multiple currencies based on CBI-time-changed L\'evy processes. The proposed framework captures the typical risk characteristics of FX markets and is coherent with the symmetries of FX rates. Moreover, due to the self-exciting behavior of CBI processes, the volatilities of FX rates exhibit self-exciting dynamics. By relying on the theory of affine processes, we show that our approach is analytically tractable and that the model structure is invariant under a suitable class of risk-neutral measures. A semi-closed pricing formula for currency options is obtained by Fourier methods.
We propose two calibration methods, also by relying on deep-learning techniques, and show that a simple specification of the model can achieve a good fit to market data on a currency triangle.
\end{abstract}

\maketitle

\section{Introduction}\label{sec:introcbitclcurrency} 
	
The Foreign-Exchange (FX) market is one of the largest in the world (see, e.g., \cite{BIS19,Woold19}). From the perspective of quantitative finance, modeling the FX market poses several challenges. First, multi-currency models must respect the symmetric structure of FX rates. To illustrate this aspect, let $S^{d,f}$ represent the value of one unit of a foreign currency $f$ measured in units of the domestic currency $d$. In a multi-currency model, the following symmetric relations must hold:
\begin{itemize}
\item $S^{f,d}=1/S^{d,f}$: the reciprocal of $S^{d,f}$ must coincide with $S^{f,d}$, representing the value of one unit of currency $d$ measured in units of currency $f$. This is referred to as \emph{inversion}.
\item $S^{d,f} = S^{d,e} \times S^{e,f}$, for any other foreign currency $e$. In other words, the FX rate $S^{d,f}$ must be inferred from $S^{d,e}$ and $S^{e,f}$ through multiplication. This is referred to as {\em triangulation}.
\end{itemize}

Besides these symmetric relations, the FX market presents some specific risk characteristics that should be properly reflected in a multi-currency model. First, FX markets are affected by stochastic volatility and jump risk, similarly to the case of equity markets. This has led to the application to FX markets of well-known models initially conceived for stock returns, such as the Heston model (see also Section \ref{sec:literature} below).
Second, the dependence between FX rates is typically stochastic and, in particular, shows evidence of unpredictable changes over time, thus generating correlation risk.
Third, the skew of the FX volatility smile exhibits a stochastic behavior. This fact has been documented in \cite{CW07} by analyzing the time series of risk-reversals\footnote{We recall that a risk-reversal, in the context of FX options, measures the difference in implied volatility between an OTM call option and a put option with the same characteristics and a symmetric delta.}, showing that their values vary significantly over time and exhibit repeated sign changes.
Finally, FX rates are affected by volatility clustering effects, similarly to many other asset classes (see, e.g., \cite[Chapter 7]{ContTankov}). This phenomenon can be observed in Figure \ref{fig:time_series}, which displays the time series of the USD-JPY exchange rate over the period 01/01/2012 - 31/12/2015.
	
	\begin{figure}[t]
		\centering
		\includegraphics[width=0.7\textwidth]{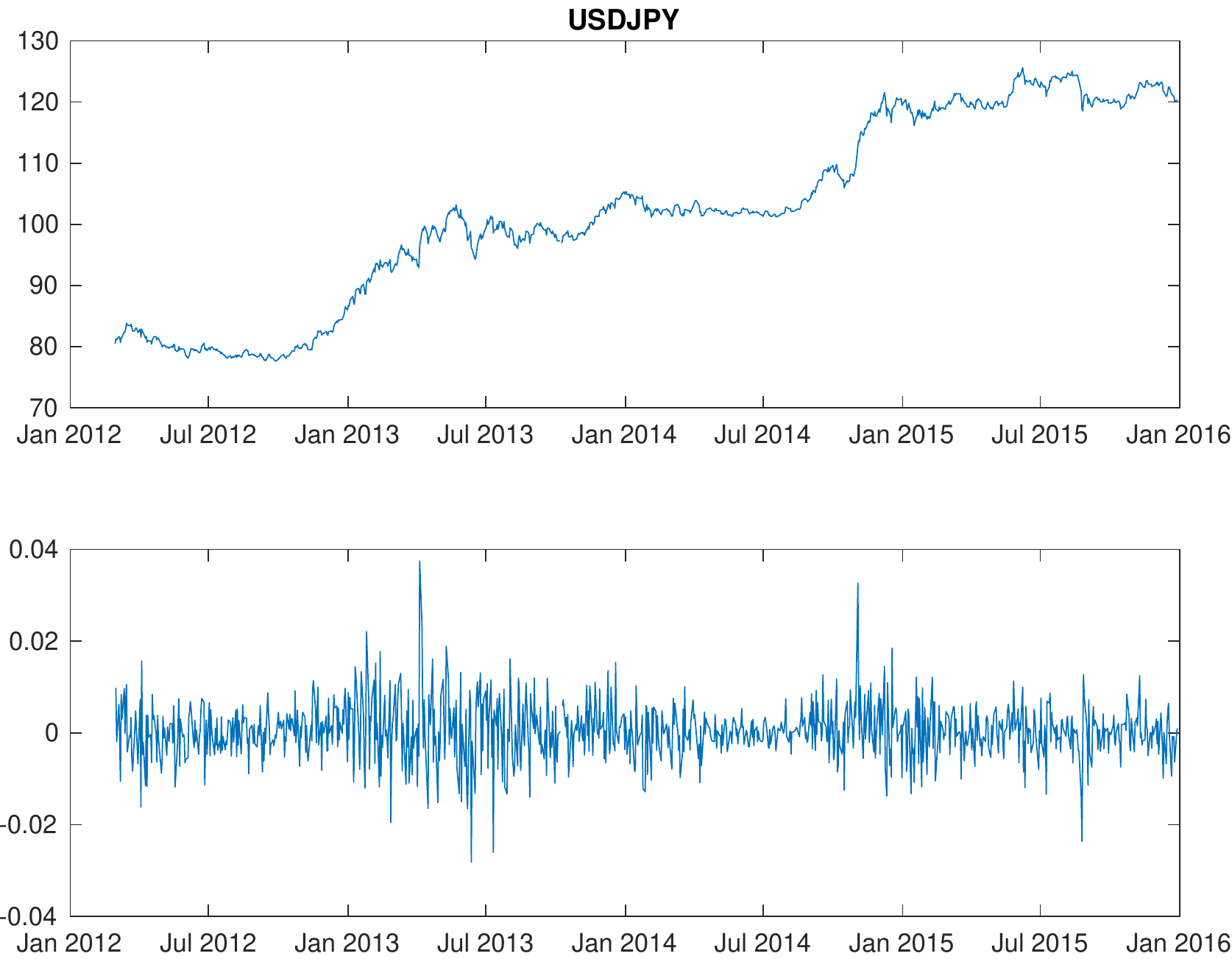}
		\caption{Time series of the USD-JPY exchange rate (upper panel) and its return (lower panel), period 01/01/2012 - 31/12/2015 (source: Bloomberg).}
		\label{fig:time_series}
	\end{figure}

In this paper, we develop a modeling framework for multiple currencies that is consistent with the symmetric structure of FX rates and captures all risk characteristics described above, including stochastic dependence among FX rates as well as between FX rates and their volatilities.
We consider models driven by {\em CBI-time-changed L\'evy processes} (CBITCL processes), a broad and flexible class of processes that allows for self-exciting jump dynamics, with stochastic volatility and mean-reversion (we refer to \cite{FGS22,phdthesis_szulda} for a thorough analysis of CBITCL processes). The proposed approach is fully analytically tractable, due to the fact that CBITCL processes are affine processes and, therefore, their characteristic function can be explicitly characterized. Moreover, we will show that CBITCL processes are {\em coherent} in the sense of \cite{Gno17}, meaning that if an FX rate is modeled by a CBITCL process, then its reciprocal also belongs to the same model class.

We construct our modeling framework by adopting an {\em artificial currency approach} (see also Section \ref{sec:literature} below), which consists in modeling each FX rate as the ratio of two primitive processes, with one primitive process associated to each currency. FX rates then satisfy the inversion and triangulation symmetries by construction and the model formulation reduces to modeling all primitive processes by means of a common family of CBITCL processes.
By relying on a Girsanov-type result for CBITCL processes, we characterize a class of risk-neutral measures that leave invariant the structure of the model. In particular, this allows preserving the CBITCL property under equivalent changes of probability, which enables us to derive an efficient pricing formula for currency options by means of Fourier techniques.

We analyze the empirical performance of a two-dimensional specification of our framework, driven by tempered $\alpha$-stable CBI processes (as recently introduced in \cite{FGS21}) and CGMY processes (see \cite{CGMY02}). We perform a calibration of the model with respect to an FX triangle consisting of three major currency pairs (USD-JPY, EUR-JPY, EUR-USD). We propose two calibration methods: a {\em standard} calibration algorithm and a {\em deep} calibration algorithm, inspired by the deep learning techniques recently developed in \cite{HMT21} and applied here for the first time in a multi-currency setting.
We assess the importance of jumps by showing that our CBITCL model achieves a superior calibration performance with respect to an analogous continuous-path model.
	
\subsection{Related literature} \label{sec:literature}

As mentioned above, our framework is based on an {\em artificial currency approach}, which goes back to the works of \cite{FH97} and \cite{Doust07} ({\em intrinsic currency approach}, in the terminology of \cite{Doust07}). Relying on this approach, several stochastic volatility models for multiple currencies have been developed, see e.g. \cite{Doust12,DeCGG13,GG14,BGP15,GGP21} in a Brownian setting. 
In the latter works, the symmetries of FX rates are respected and several sources of risk can be adequately represented, with the important exceptions of jump risk, volatility clustering and self-excitation effects, which will play an important role in our framework.

Most of the works mentioned in the previous paragraph can be regarded as multi-currency extensions of the Heston model (to this effect, see also \cite{JKWW11}). This is motivated by the fact that the Heston model is known to be stable under inversion (see, e.g., \cite{DBR08}). The property that a certain model class is invariant under inversion has been termed {\em coherence} in \cite{Gno17}, where it has been shown that general affine stochastic volatility models are coherent.  As will be shown below, our modeling framework is coherent in this sense. The coherence property has been recently studied by \cite{GBP20} (under the name of {\em consistency}) in the class of jump-diffusion models.

Beyond the Brownian setting, some important contributions to the modeling of FX rates include the work of \cite{EK06} based on time-inhomogeneous L\'evy processes and the work of \cite{CW07} based on time-changed L\'evy processes with CIR-type activity rates.
More recently, \cite{BDR17} have developed a multi-currency  framework driven by a multi-dimensional L\'evy process with dependent components. In this work, only standard L\'evy processes are considered and, therefore, stochastic volatility is not explicitly modeled. \cite{BDR17} suggest however the possible use of time change methods. This line of research is pursued in \cite{BM18} and further expanded in our work. In particular, \cite{BM18} propose a model based on time-changed L\'evy processes where the activity rate can have jumps and self-excitation. The presence of common jumps between the L\'evy process and the activity rate induces a non-trivial dependence between the FX rate and its volatility. The model of \cite{BM18} is however limited to a single FX rate, in which case the FX symmetries do not play any role. In contrast, we propose a coherent multi-currency framework that satisfies the FX symmetries and exhibits rich and flexible stochastic dynamics for all FX rates and their volatilities, while preserving an analytical tractability comparable to the model of \cite{BM18}.
	
\subsection{Outline of the paper} 
	
The paper is structured as follows. In Section \ref{sec:cbitcl} we recall some basic results on CBITCL processes. In Section \ref{sec:generalmodelcbitclcurrency} we develop the multi-currency modeling framework, while in Section \ref{sec:numericalcbitclcurrency} we  present two calibration methods and analyze the empirical fit to market data on an FX triangle. Finally, Section \ref{sec:conclusioncbitclcurrency} concludes the paper.

\section{CBI-time-changed L\'evy processes}
\label{sec:cbitcl}

In this section, we present some fundamental results on {\em CBI-time-changed L\'evy (CBITCL) processes}, referring to \cite{FGS22} for a complete theoretical analysis of this class of processes, detailed proofs and additional results (see also \cite[Chapter 4]{phdthesis_szulda}). We work on a filtered stochastic basis $(\Omega,\cF,\FF,\QQ)$, where $\FF$ is a filtration satisfying the usual conditions. 

\subsection{Definition and characterization of CBITCL processes} 

Let us start by recalling the definition of a {\em Continuous-state Branching processes with Immigration} (see \cite[Section 5]{Li20}).
\begin{itemize}
\item Let the function $\Psi:\R_-\rightarrow\R$ be given by
\be\label{eq:immigrationmechanism}
\Psi(x) := \beta\,x + \int_0^{+\infty} {(e^{\,x\,z} - 1)\,\nu(\ud z)}, \qquad \forall x \in\R_-,
\ee
where $\beta \geq0$ and $\nu$ is a L\'evy measure on $(0,+\infty)$ such that $\int_0^1 {z\,\nu(\ud z)}<+\infty$;	
\item Let the function $\Phi:\R_-\rightarrow\R$ be given by 
\be\label{eq:branchingmechanism}
\Phi(x) := -\,b\,x +  \frac{\sigma^2}{2}\,x^2 + \int_0^{+\infty} {(e^{\,x\,z} - 1 - x\,z)\,\pi(\ud z)}, \qquad \forall x \in\R_-,
\ee
where $b \in \R$, $\sigma \in\R$ and $\pi$ is a L\'evy measure on $(0,+\infty)$ such that $\int_1^{+\infty} {z\,\pi(\ud z)}<+\infty$.	
\end{itemize}

\begin{defn}\label{def:defcbi}
A Markov process $X=(X_t)_{t \geq 0}$ with initial value $X_0$ and state space $\R_+$ is said to be a \emph{Continuous-state Branching process with Immigration} (CBI) with immigration mechanism $\Psi$ and branching mechanism $\Phi$ if its Laplace transform is given by
\[
\EE[e^{uX_T}] = \exp\left(\int_0^T {\Psi\bigl(\cV(s,u)\bigr)\,\ud s} + \cV(T,u)X_0\right),
\]
for all $u \in\R_-$ and $T \in\R_+$, where the function $\cV(\cdot,u):\R_+\rightarrow\R_-$ is the unique solution to 
\[
\frac{\partial\cV}{\partial t}(t,u) = \Phi\bigl(\cV(t,u)\bigr), \qquad \cV(0,u) = u.
\]
\end{defn}

Definition \ref{def:defcbi} corresponds to a conservative stochastically continuous CBI process in the sense of \cite{KW71}. Note that CBI processes are non-negative, strongly Markov (Feller) and with c\`adl\`ag trajectories. As a consequence, the path integral $Y:=\int_0^{\cdot}X_s\,\ud s$ of a CBI process $X$ is always well defined as a non-decreasing process. It can therefore be used as a finite continuous time-change. This motivates the following definition.

\begin{defn}\label{def:cbitcl}
A process $(X,Z)=((X_t, Z_t))_{t \geq 0}$ is said to be a \emph{CBI-time-changed L\'evy process} (CBITCL process) if
\begin{enumerate}
\item[(i)] $X$ is a CBI process, and
\item[(ii)] $Z=L_Y$, where $L=(L_t)_{t\geq0}$ is a L\'evy process independent of $X$ and $Y=(Y_t)_{t\geq0}$ denotes the process defined by $Y_t:=\int_0^tX_s\,\ud s$, for all $t\in\R_+$.
\end{enumerate}
\end{defn}

The L\'evy exponent $\Xi$ of the L\'evy process $L$ admits the L\'evy-Khintchine representation
\[
\Xi(u) := b_Z\,u + \frac{\sigma_Z^2}{2}\,u^2 + \int_{\R\setminus\{0\}} {\bigl(e^{zu} - 1 - zu\ind_{\{|z|<1\}}\bigr)\,\gamma_Z(\ud z)}, 
\qquad \forall u \in \im\R,
\]
where $(b_Z, \sigma_Z, \gamma_Z)$ is the L\'evy triplet of $L$, with $b_Z \in \R$, $\sigma_Z \in\R$ and $\gamma_Z$ a L\'evy measure on $\R$. 
In the following, we shall write $\uCBITCL(X_0, \Psi, \Phi, \Xi)$ to denote that a process $(X,Z)$ is a CBI-time-changed L\'evy process in the sense of Definition \ref{def:cbitcl}, where $\Psi$ and $\Phi$ denote respectively the immigration and branching mechanisms of the CBI process $X$ and $\Xi$ the L\'evy exponent of $L$.

CBI processes admit a characterization in terms of a Lamperti representation (see  \cite{EPU13} and also \cite[Chapter 2]{phdthesis_szulda}). However, it turns out that there exists an equivalent representation that is better suited to our purposes, in terms of solutions to certain stochastic integral equations of the Dawson-Li type (see \cite{DL06}). To this effect, let us introduce the following objects:
\begin{itemize}
\item two Brownian motions $B^1=(B_t^1)_{t \geq 0}$ and $B^2=(B_t^2)_{t \geq 0}$;
\item a Poisson random measure $N_0(\ud t, \ud x)$ on $(0,+\infty)^2$ with compensator $\ud t\,\nu(\ud x)$ and compensated measure $\tildeN_0(\ud t, \ud x) := N_0(\ud t, \ud x) - \ud t\,\nu(\ud x)$;
\item a Poisson random measure $N_1(\ud t, \ud u, \ud x)$ on $(0,+\infty)^3$ with compensator $\ud t\,\ud u\,\pi(\ud x)$ and compensated measure $\tildeN_1(\ud t, \ud u, \ud x) := N_1(\ud t, \ud u, \ud x) - \ud t\,\ud u\,\pi(\ud x)$;
\item a Poisson random measure $N_2(\ud t, \ud u, \ud x)$ on $(0,+\infty)^2\times\R$ with compensator $\ud t\,\ud u\,\gamma_Z(\ud x)$ and compensated measure $\tildeN_2(\ud t, \ud u, \ud x) := N_2(\ud t, \ud u, \ud x) - \ud t\,\ud u\,\gamma_Z(\ud x)$.	
\end{itemize}
We furthermore assume that $B^1$, $B^2$, $N_0$, $N_1$ and $N_2$ are mutually independent. 
For any $X_0\in\R_+$, let us consider the following stochastic integral equations: 
\begin{align}
X_t &= X_0 + \int_0^t {\bigl( \beta - b\,X_s \bigr)\,\ud s} + \sigma\int_0^t {\sqrt{X_s}\,\ud B_s^1} \notag \\
&\quad+ \int_0^t \int_0^{+\infty} {x\,N_0(\ud s, \ud x)} + \int_0^t \int_0^{X_{s-}}\! \int_0^{+\infty} {x\,\tildeN_1(\ud s, \ud u, \ud x)},
\label{eq:cbitclsde1}\\
Z_t &= b_Z\int_0^t {X_s\,\ud s} + \sigma_Z\int_0^t {\sqrt{X_s}\,\ud B_s^2} + \int_0^t \int_0^{X_{s-}}\!\int_{|x| \geq 1} {x\,N_2(\ud s, \ud u, \ud x)}\notag\\
&\quad+ \int_0^t \int_0^{X_{s-}}\!\int_{|x| < 1} {x\,\tildeN_2(\ud s, \ud u, \ud x)}.
\label{eq:cbitclsde2}
\end{align}

The connection between Definition \ref{def:cbitcl} and the stochastic integral equations \eqref{eq:cbitclsde1}-\eqref{eq:cbitclsde2} is given in the following proposition (see \cite[Theorem 2.3]{FGS22}). 

\begin{prop}\label{lem:weakequivalencecbitcl}
A process $(X,Z)$ with initial value $(X_0, 0)$ is a $\uCBITCL(X_0, \Psi, \Phi, \Xi)$ process if and only if it is a weak solution to the stochastic integral equations \eqref{eq:cbitclsde1}-\eqref{eq:cbitclsde2}.
\end{prop}

On a given stochastic basis $(\Omega,\cF,\FF,\QQ)$, there exists a unique strong solution to \eqref{eq:cbitclsde1}-\eqref{eq:cbitclsde2}. Indeed, there is a unique strong solution $X=(X_t)_{t \geq 0}$ to  \eqref{eq:cbitclsde1}, which corresponds to the \emph{Dawson-Li representation} of a CBI process  (see \cite[Theorems 5.1 and 5.2]{DL06}). 
In turn, since the right-hand side of \eqref{eq:cbitclsde2} does depend only on the process $X=(X_t)_{t \geq 0}$ and not on $Z=(Z_t)_{t\geq0}$, this obviously implies the existence of a unique strong solution $Z=(Z_t)_{t\geq0}$ to \eqref{eq:cbitclsde2} as well.
In the following, if a CBITCL process $(X,Z)$ is directly defined as the unique strong solution to \eqref{eq:cbitclsde1}-\eqref{eq:cbitclsde2}, we will say that $(X,Z)$ is defined through its  \emph{extended Dawson-Li representation} (see \cite[Section 2.1]{FGS22}). 

\begin{rem}
The system of stochastic integral equations \eqref{eq:cbitclsde1}-\eqref{eq:cbitclsde2} makes evident the self-exciting behavior of a CBITCL process. More specifically, we can observe the following:
\begin{itemize}
\item For the CBI process $X$, the domain of integration of the integral with respect to $\widetilde{N}_1$ depends on the value of the process itself. This generates a self-exciting effect since, when a jump occurs, the jump intensity of $X$ increases. In turn, this increases the likelihood of subsequent jumps of $X$, thereby generating jump clustering phenomena.
\item The volatility components of $Z$ depend on the value of the process $X$. Therefore, large values of $X$ increase the volatility of the process $Z$, thereby increasing the likelihood of volatility clusters in the dynamics of $Z$ as well as joint clusters between $X$ and $Z$.
\end{itemize}
\end{rem}

\subsection{Affine property and changes of probability}

The next proposition shows that CBITCL processes are affine and provides an explicit characterization of the Laplace-Fourier transform. 

\begin{prop}\label{prop:cbitclaffine}  
Let $(X,Z)$ be a $\uCBITCL(X_0, \Psi, \Phi, \Xi)$ process and consider the joint process $(X,Y,Z)$, where $Y:= \int_0^{\cdot}{X_s\,\ud s}$. Then, the process $(X,Y,Z)$ is an affine process on the state space $\R_+^2\times\R$ with conditional Laplace-Fourier transform given by 
\[
\EE\bigl[e^{\,u_1\,X_T + u_2\,Y_T + u_3\,Z_T}\,\bigr|\,\cF_t\bigr] = \exp\Bigl(\cU(T-t,u_1, u_2, u_3) + \cV(T-t, u_1, u_2, u_3)\,X_t + u_2\,Y_t + u_3\,Z_t\Bigr), 
\]
for all $(u_1, u_2, u_3) \in \C_-^2\times\im\R$ and $0 \leq t \leq T < +\infty$, where the functions $\cU(\cdot, u_1, u_2, u_3):\R_+\rightarrow\C$ and $\cV(\cdot, u_1, u_2, u_3):\R_+\rightarrow\C_-$ are solutions to
\begin{align}
\cU(t, u_1, u_2, u_3) &= \int_0^{t} {\Psi\bigl(\cV(s, u_1, u_2, u_3)\bigr)\,\ud s},\label{eq:Riccati1}\\
\frac{\partial \cV}{\partial t}(t, u_1, u_2, u_3) &= \Phi\bigl(\cV(t, u_1, u_2, u_3)\bigr) + u_2 + \Xi(u_3), \qquad \cV(0, u_1, u_2, u_3) = u_1,
\label{eq:Riccati2}
\end{align}	
where $\Psi:\C_-\rightarrow\C$ and $\Phi:\C_-\rightarrow\C$ denote  the analytic extensions to $\C_-$ of the corresponding functions defined in \eqref{eq:immigrationmechanism} and \eqref{eq:branchingmechanism}, respectively.
\end{prop}
\begin{proof}
By \cite[Corollary 2.10]{DFS03}, the process $X$ is an affine process. The affine property of $(X,Y,Z)$ and the characterization of its conditional Laplace-Fourier transform then follow by an application of \cite[Theorems 4.10 and 4.16]{phdthesis_kr} (see \cite[Section 2.2]{FGS22} for additional details).
\end{proof}

\begin{rem}\label{rem:coherence}
In view of Proposition \ref{prop:cbitclaffine}, CBITCL processes can be viewed as {\em affine stochastic volatility models} in the sense of \cite[Chapter 5]{phdthesis_kr}. In the context of FX modeling, such models have been proved to be {\em coherent} by \cite{Gno17}, as mentioned in the Introduction. This fact will play an important role in the construction of our multi-currency framework in Section \ref{sec:generalmodelcbitclcurrency}.
We refer to \cite[Section 2.2]{FGS22} for a detailed analysis of the relation between CBITCL processes and affine stochastic volatility models.
\end{rem}

We close this section by describing a class of equivalent changes of probability of Esscher type that leave invariant the CBITCL structure. Let us first define the convex set $\cD_X$ as follows:
\be\label{eq:domaincbi}
\cD_X := \biggl\{ x \in \R : \int_1^{+\infty} {e^{xz}(\nu+\pi)(\ud z)} < +\infty \biggr\}.
\ee
The set $\cD_X$ is the effective domain of the functions $\Psi$ and $\Phi$, which can be extended as finite-valued convex functions on $\cD_X$. Note that $\cD_X$ also represents the extended domain of the Laplace transform of the CBI process $X$ (see \cite[Theorem 2.6]{FGS21}). Let us also introduce the convex set
\be\label{eq:domaincbitcl}
\cD_Z := \biggl\{ x \in \R : \int_{|z| \geq 1} {e^{xz}\gamma_Z(\ud z)} < +\infty \biggr\},
\ee
which represents the effective domain of the L\'evy exponent $\Xi$ when restricted to real arguments. Let us fix $\zeta \in \R$ and $\lambda \in \R$ and consider the process $\cW=(\cW_t)_{t \geq 0}$ defined by
\be\label{eq:processW}
\cW_t := \zeta\,(X_t-X_0) + \lambda\,Z_t, \qquad \text{ for all } t\in\R_+.
\ee 
By \cite[Proposition II.8.26]{JS03}, it can be checked that $\cW$ is an exponentially special semimartingale if and only if $\zeta \in \cD_X$ and $\lambda \in \cD_Z$. In this case, $\cW$ admits a unique exponential compensator, i.e., a predictable process of finite variation, denoted by $\cK=(\cK_t)_{t \geq 0}$, such that $\exp(\cW - \cK)$ is a local martingale (see \cite{KS02cumulant}). 
The following lemma provides the explicit expression of $\cK$.

\begin{lem}\label{lem:exponentialcompensator}
Let $(X,Z)$ be a $\uCBITCL(X_0, \Psi, \Phi, \Xi)$ process. Consider the process $\cW$ defined in \eqref{eq:processW}, with $\zeta \in \cD_X$ and $\lambda \in \cD_Z$. Then, the exponential compensator $\cK$ of $\cW$ is given by 
\be\label{eq:exponentialcompensator}
\cK_t = t\,\Psi(\zeta) + Y_t\,\bigl(\Phi(\zeta) + \Xi(\lambda)\bigr), \qquad \text{for all }\,t\in\R_+.
\ee
\end{lem}
\begin{proof}
For brevity of presentation, we only give a sketch of the proof, referring to \cite[Lemma 4.1]{FGS22} for full details. Since CBITCL processes are quasi-left-continuous, the exponential compensator $\cK$ can be explicitly computed in terms of the semimartingale differential characteristics of $(X,Z)$ (see \cite{KS02cumulant}). In view of Proposition \ref{lem:weakequivalencecbitcl}, the differential semimartingale characteristics of $(X,Z)$ can be easily obtained from \eqref{eq:cbitclsde1}-\eqref{eq:cbitclsde2}. Representation \eqref{eq:exponentialcompensator} then follows by standard computations.
\end{proof}

Fixing a time horizon $\horizon < +\infty$, we can state the following Girsanov-type result for CBITCL processes, which will play a central role in the modeling framework developed in the next section.
In the following statement, we denote by $\cD^{\circ}_X$ and $\cD_Z^{\circ}$ the interior of sets $\cD_X$ and $\cD_Z$, respectively.

\begin{thm}\label{thm:girsanovcbitcl}
Let $(X,Z)$ be a $\uCBITCL(X_0, \Psi, \Phi, \Xi)$ process. Consider the process $\cW$ defined in \eqref{eq:processW}, with $\zeta \in \cD^{\circ}_X$ and $\lambda \in \cD^{\circ}_Z$, and its exponential compensator $\cK$ given by \eqref{eq:exponentialcompensator}. Then, the  process $(\exp(\cW_t - \cK_t))_{t\in[0,\cT]}$ is a martingale.
Moreover, setting
\be\label{eq:measurechangecbitcl}
\frac{\ud\QQ^{\prime}}{\ud\QQ} := e^{\cW_{\cT} - \cK_{\cT}},
\ee 
defines a probability measure $\QQ^{\prime}\sim\QQ$ under which $(X,Z)$ remains a CBITCL process up to time $\cT$ with parameters $\beta'$, $\nu'$, $b'$, $\sigma'$, $\pi'$, $b'_Z$, $\sigma'_Z$ and $\gamma_Z'$ given in Table \ref{table:newparameterscbitcl}.
\end{thm}
\begin{proof}
By definition of the exponential compensator, the process $\exp(\cW-\cK)$ is a strictly positive local martingale. The true martingale property of the process $\exp(\cW_t-\cK_t)_{t\in[0,\cT]}$ follows similarly as in  \cite[Theorem 3.2]{KRM15}, since $\zeta \in \cD^{\circ}_X$ and $\lambda \in \cD^{\circ}_Z$. The fact that $(X,Z)$ is a CBITCL process up to time $\cT$ with parameters given in Table \ref{table:newparameterscbitcl} under $\QQ^{\prime}$ is a consequence of Girsanov's theorem together with Proposition \ref{lem:weakequivalencecbitcl} (see \cite[Theorem 4.2]{FGS22} for full details).
\end{proof}

\begin{table}[h!] 
	\centering
	\captionsetup{justification=centering}
	\begin{tabular}{|l|}
		\hline
		CBITCL parameters under $\QQ^{\prime}$ \\
		\hline
		\hline
		$\beta^{\prime} := \beta$ \\
		\hline
		$\nu^{\prime}(\ud z) := e^{\zeta z}\nu(\ud z)$ \\
		\hline
		$b^{\prime} := b - \zeta\,\sigma^2 - \int_0^{+\infty} {z(e^{\zeta z}-1)\pi(\ud z)}$ \\ 
		\hline
		$\sigma^{\prime} := \sigma$ \\
		\hline
		$\pi^{\prime}(\ud z) := e^{\zeta z}\pi(\ud z)$ \\
		\hline
		$b_Z^{\prime} := b_Z + \lambda\,\sigma_Z^2 + \int_{|z|<1} {z(e^{\lambda z}-1)\,\gamma_Z(\ud z)}$ \\
		\hline
		$\sigma_Z^{\prime} := \sigma_Z$ \\
		\hline
		$\gamma_Z^{\prime}(\ud z) := e^{\lambda z}\gamma_Z(\ud z)$ \\
		\hline
	\end{tabular}
	\caption{Parameter transformations from $\QQ$ to $\QQ^{\prime}$\\ for the CBITCL process $(X,Z)$.}
	\label{table:newparameterscbitcl}
\end{table}

\section{Modeling of multiple currencies via CBITCL processes}\label{sec:generalmodelcbitclcurrency}

In this section, we present our modeling framework for a multi-currency market. In Section \ref{sec:fxmarket}, we introduce the main quantities to be modeled together with the specific requirements induced by absence of arbitrage and by the FX symmetries discussed in the Introduction. Section \ref{sec:constructionmodelcbitclcurrency} contains the construction of the framework and the description of its most relevant features. In Section \ref{sec:currencypricing}, Fourier techniques are applied to the pricing of currency options. 
We work on a filtered stochastic basis $(\Omega,\cF,\FF,\QQ)$ and consider models defined up to a time horizon $\cT<+\infty$.

\subsection{Definition of the multiple currency market}\label{sec:fxmarket}

The FX market involves different economies, each of them associated to a specific currency. The $i\textsuperscript{th}$ and $j\textsuperscript{th}$ currencies are related by the spot FX rate process $S^{i,j}$, representing the value of one unit of currency $j$ measured in units of currency $i$. Our definition of the multiple currency market will make use of the following ingredients, denoting by $N\in\N$, with $N\geq2$, the number of economies (i.e., currencies) considered:
\begin{enumerate}
\item[(i)] $\bm{D} = \{D^i;i=1,\ldots,N\}$ is an $\R^N_{>0}$-valued  process, with $D^i$ representing the bank account of the $i\textsuperscript{th}$ economy, for $i=1,\ldots,N$;
\item[(ii)] $\bm{S} = \{S^{i,j};i,j=1,\ldots,N\}$ is an $\R^{N \times N}_{>0}$-valued  process representing the spot FX rates between the $N$ different currencies and such that $S^{i,i}\equiv 1$, for all $i=1,\ldots,N$.
\end{enumerate}

\begin{defn} \label{def:fxmarket}
We say that the pair $(\bm{D},\bm{S})$ represents a {\em multiple currency market} if, for every $i=1,\ldots,N$, the following assets are traded in the $i\textsuperscript{th}$ economy:
\begin{itemize}
\item the bank account $D^i$;
\item for every $j=1,\ldots,N$ with $j \neq i$, the bank account of the $j\textsuperscript{th}$ economy denominated in units of the $i\textsuperscript{th}$ currency, namely $S^{i,j}D^j$.
\end{itemize}
\end{defn}

We aim at constructing models for multiple currency markets that respect the FX symmetries mentioned in the Introduction and satisfy absence of arbitrage in the sense of {\em no free lunch with vanishing risk} (NFLVR). Since $(\bm{D},\bm{S})$ are strictly positive processes, \cite[Theorem 1.1]{DS98} implies that, for each $i=1,\ldots,N$, NFLVR holds in the $i\textsuperscript{th}$ economy if and only if there exists a risk-neutral measure $\QQ^i$, i.e., a probability measure $\QQ^i$ equivalent to $\QQ$ such that $S^{i,j}D^j/D^i$ is a local martingale under $\QQ^i$.
We then formulate the following definition, which extends \cite[Definition 1]{EG18} to an FX market consisting of an arbitrary number $N$ of currencies.
	
\begin{defn} \label{def:wellposedness}
The multiple currency market $(\bm{D}, \bm{S})$ is said to be {\em well-posed} if the following hold:
\begin{enumerate}
\item[(i)] {\em no direct arbitrage}: $S^{j,i}=1/S^{i,j}$, for all $i,j=1,\ldots,N$;
\item[(ii)] {\em no triangular arbitrage}: $S^{i,j}=S^{i,k} \times S^{k,j}$, for all $i,k,j=1,\ldots,N$; 
\item[(iii)] there exists a risk-neutral measure $\QQ^i$ for the $i\textsuperscript{th}$ economy, for all $i=1,\ldots,N$.
\end{enumerate}
\end{defn}

Besides the requirement of well-posedness, we are interested in multi-currency models that are {\em coherent} in the sense of \cite{Gno17}. This means that, if the FX rate process $S^{i,j}$ belongs to a certain model class under $\QQ^i$, then also its reciprocal $S^{j,i}$ belongs to the same model class under $\QQ^j$, where $\QQ^i$ and $\QQ^j$ are risk-neutral measures for the $i\textsuperscript{th}$ and $j\textsuperscript{th}$ economy, respectively. Obviously, {\em coherence} is a desirable property from the modeling perspective, since it ensures that the model retains its analytical tractability in all $N$ different economies.

To achieve a well-posed as well as coherent multiple currency market, we will proceed as follows:
\begin{enumerate}
\item 
Adopting the {\em artificial currency approach}, we express each currency with respect to an artificial currency indexed by $0$ and construct the artificial FX rates $S^{0,i}$, for $i=1,\ldots,N$. 
\item
We define the FX rates $S^{i,j}$, for all $i,j=1,\ldots,N$, by taking suitable ratios of the artificial FX rates $S^{0,i}$, $i=1,\ldots,N$. Parts (i)-(ii) of Definition \ref{def:wellposedness} are then satisfied by construction.
\item
By relying on Theorem \ref{thm:girsanovcbitcl}, we construct a  risk-neutral measure $\QQ^i$, for every $i=1,\ldots,N$, under which the driving processes remain CBITCL processes, thus ensuring part (iii) of Definition \ref{def:wellposedness} as well as the stability of the structure of the model (coherence).
\end{enumerate}

\subsection{Construction of the modeling framework}\label{sec:constructionmodelcbitclcurrency}
	
The construction of our modeling framework starts by modeling the $N$ artificial FX rates $S^{0,i}$, for $i=1,\ldots,N$, by means of a common family of CBITCL processes. To this effect, we assume that the stochastic basis $(\Omega,\cF,\FF,\QQ)$ supports $d$ mutually independent CBITCL processes $(X^k,Z^k)$,  $k=1,\ldots,d$, defined through the corresponding extended Dawson-Li representations \eqref{eq:cbitclsde1}-\eqref{eq:cbitclsde2}.\footnote{In the following, we shall use the superscript $k$ to denote all parameters and driving processes appearing in the extended Dawson-Li representation \eqref{eq:cbitclsde1}-\eqref{eq:cbitclsde2} of the CBITCL process $(X^k,Z^k)$ on $(\Omega,\cF,\FF,\QQ)$, for $k=1,\ldots,d$.}
For each $k=1,\ldots,d$, we denote by $\cD_{X^k}$ and $\cD_{Z^k}$ the sets \eqref{eq:domaincbi} and \eqref{eq:domaincbitcl}, respectively, associated to the CBITCL process $(X^k,Z^k)$.

For each $i=1,\ldots,N$, we introduce the following parameters:
\begin{itemize}
\item $r^i\in\R$, representing the risk-free short rate in the $i\textsuperscript{th}$ economy and generating the bank account $D^i_t:=\exp(r^i\,t)$, for all $t\in[0,\cT]$;
\item $\zeta^i=(\zeta^i_1,\ldots,\zeta^i_d)\in\R^d$ such that $\zeta^i_k\in\cD^{\circ}_{X^k}$, for all $k=1,\ldots,d$;
\item $\lambda^i=(\lambda^i_1,\ldots,\lambda^i_d)\in\R^d$ such that $\lambda^i_k\in\cD^{\circ}_{Z^k}$, for all $k=1,\ldots,d$.
\end{itemize}
In addition, for each $i=1,\ldots,N$ and $k=1,\ldots,d$, we denote by $\cK^{i,k}=(\cK^{i,k}_t)_{t\in[0,\cT]}$ the exponential compensator of the process $(\zeta^i_k\,X^k + \lambda^i_k\,Z^k)$, as characterized in Lemma \ref{lem:exponentialcompensator}.

\begin{rem}
We point out that the modeling framework developed in this section can be easily generalized to the case of stochastic interest rates. In particular, by allowing the interest rates $r^i$, for $i=1,\ldots,N$, to be driven by the common family of CBITCL processes $(X^k,Z^k)$, $k=1,\ldots,d$, one can introduce dependence between the interest rates, the FX rates and their volatilities.
\end{rem}

For each $i=1,\ldots,N$, we specify the artificial FX rate $S^{0,i}$ as follows:
\be\label{eq:artificialfxrate}
S_t^{0,i} := S_0^{0,i}\,e^{-\,r^i\,t}\prod_{k=1}^{d} e^{\,\zeta^i_k\,X_t^k + \lambda^i_k\,Z_t^k - \,\cK_t^{i,k}}, \qquad \text{ for all } t \in[0,\cT].
\ee 

The artificial FX rates are modeling quantities that cannot be observed in reality. However, the parameters $\zeta^i_k$ and $\lambda^i_k$ will have a specific role in the dynamics of the actual FX rates. Indeed, $\lambda^i_k$ will measure the relative importance of the risk arising from the $k\textsuperscript{th}$ time-changed L\'evy process $Z^k$, while $\zeta^i_k$ will measure the dependence between the $k\textsuperscript{th}$ CBI process $X^k$ and the $i\textsuperscript{th}$ FX rate, as will become clear from Lemma \ref{lem:sdefxrate} below and the following discussion.

\begin{lem}\label{lem:generalmulticurrencymodel}
For each $i=1,\ldots,N$, the process $S^{0,i}=(S^{0,i}_t)_{t\in[0,\cT]}$ satisfies the following dynamics:
\begin{align}
\frac{\ud S_t^{0,i}}{S_{t-}^{0,i}} &= -\,r^i\,\ud t + \sum_{k=1}^{d} \left( \sqrt{X_t^k}\,\bigl( \zeta^i_k\,\sigma^k\,\ud B_t^{k,1} + \lambda^i_k\,\sigma_Z^k\,\ud B_t^{k,2} \bigr) +  \int_0^{+\infty} {(e^{\zeta^i_k x}-1)\tildeN_0^k(\ud t,\ud x)} \right)\notag\\
&\qquad\qquad + \sum_{k=1}^{d} \int_0^{X_{t-}^k}\!\left( \int_0^{+\infty} {(e^{\zeta^i_k x}-1)\tildeN_1^k(\ud t,\ud u,\ud x)} + \int_{\R} {(e^{\lambda^i_k x}-1)\tildeN_2^k(\ud t,\ud u,\ud x)}\right).\label{eq:sdeartificialfxrate}
\end{align}
Moreover, the process $S^{0,i}D^i=(S^{0,i}_tD^i_t)_{t\in[0,\cT]}$ is a martingale on $(\Omega,\cF,\FF,\QQ)$, for all $i=1,\ldots,N$.
\end{lem}
\begin{proof}
Using specification \eqref{eq:artificialfxrate} of $S^{0,i}$, equation \eqref{eq:sdeartificialfxrate} follows from an application of It\^o's formula together with the extended Dawson-Li representation \eqref{eq:cbitclsde1}-\eqref{eq:cbitclsde2} of the process $(X^k,Z^k)$, for $k=1,\ldots,d$.
The martingale property of $S^{0,i}D^i$ follows from the independence of the CBITCL processes $(X^k,Z^k)$, for  $k=1,\ldots,d$, together with the martingale property stated in Theorem \ref{thm:girsanovcbitcl}. 
\end{proof}

We define the FX rate process $\bm{S}=\{S^{i,j};i,j=1,\ldots,N\}$ as follows:
\be\label{eq:fxrate}
S_t^{i,j} := \frac{S_t^{0,j}}{S_t^{0,i}}, \qquad \text{ for all } i,j=1,\ldots,N \text{ and } t\in[0,\cT].
\ee
This specification of $\bm{S}$ ensures that the inversion and triangulation symmetries of FX rates (corresponding respectively to parts (i) and (ii) of Definition \ref{def:wellposedness}) are satisfied by construction, thereby completing steps (1) and (2) of the model construction outlined at the end of Section \ref{sec:fxmarket}.

The next corollary describes a class of probability measures that leave invariant the structure of our multi-currency model driven by CBITCL processes. This result plays a crucial role in ensuring absence of arbitrage and coherence of our  framework.

\begin{cor}\label{thm:girsanovcbitclcurrency}
For each $i=1,\ldots,N$, setting
\be\label{eq:measurechangecbitclcurrency}
\frac{\ud\QQ^i}{\ud\QQ} := \frac{S_{\cT}^{0,i}\,D_{\cT}^i}{S_0^{0,i}},
\ee 
defines a probability measure $\QQ^i\sim\QQ$ under which $(X^k,Z^k)$, for $k=1,\ldots,d$, remain mutually independent CBITCL processes (up to time $\cT$) with parameters given in Table \ref{table:newparameterscbitclcurrency}.
Moreover, for each $i=1,\ldots,N$, the probability measure $\QQ^i$ is a risk-neutral measure for the $i\textsuperscript{th}$ economy.
\end{cor}
\begin{proof}
In view of Lemma \ref{lem:generalmulticurrencymodel}, $\QQ^i$ is well-defined by \eqref{eq:measurechangecbitclcurrency} as a probability measure equivalent to $\QQ$, for each $i=1,\ldots,N$. The fact that $(X^k,Z^k)$, for all $k=1,\ldots,d$, remains a CBITCL process under $\QQ^i$ with parameters given in Table \ref{table:newparameterscbitclcurrency} follows from Theorem \ref{thm:girsanovcbitcl} together with the independence of the processes $(X^k,Z^k)$, for $k=1,\ldots,d$. Moreover, again the independence of the processes $(X^k,Z^k)$, for $k=1,\ldots,d$, under $\QQ$ together with the structure of the probability $\QQ^i$ defined in \eqref{eq:measurechangecbitclcurrency} implies that the mutual independence is preserved under $\QQ^i$.
Finally, for all $i,j=1,\ldots,N$, in view of \eqref{eq:fxrate} and \eqref{eq:measurechangecbitclcurrency}, the process $S^{i,j}D^j/D^i$ is a local martingale under $\QQ^i$ if and only if $S^{0,j}D^j$ is a local martingale under $\QQ$. Since the latter property holds by Lemma \ref{lem:generalmulticurrencymodel}, the proof is complete.
\end{proof}

\begin{table}[ht]
		\centering
		\captionsetup{justification=centering}
		\begin{tabular}{|l|}
			\hline
			CBITCL parameters under $\QQ^i$ \\
			\hline
			\hline
			$\beta^{i,k} := \beta^k$ \\
			\hline
			$\nu^{i,k}(\ud z) := e^{\zeta^i_k z}\nu^k(\ud z)$ \\
			\hline
			$b^{i,k} := b^k - \zeta^i_k\,(\sigma^k)^2 - \int_0^{+\infty} {z(e^{\zeta^i_k z}-1)\pi^k(\ud z)}$ \\ 
			\hline
			$\sigma^{i,k} := \sigma^k$ \\
			\hline
			$\pi^{i,k}(\ud z) := e^{\zeta^i_k z}\pi^k(\ud z)$ \\
			\hline
			$b_Z^{i,k} := b_Z^k + \lambda^i_k\,(\sigma_Z^k)^2 + \int_{|z|<1} {z(e^{\lambda^i_k z}-1)\gamma_Z^k(\ud z)}$ \\
			\hline
			$\sigma_Z^{i,k} := \sigma_Z^k$ \\
			\hline
			$\gamma_Z^{i,k}(\ud z) := e^{\lambda^i_k z}\gamma_Z^k(\ud z)$ \\
			\hline
		\end{tabular}
		\caption{Parameter transformations from $\QQ$ to $\QQ^i$\\ for the CBITCL process $(X^k, Z^k)$.}
		\label{table:newparameterscbitclcurrency}
\end{table}

\begin{rem}
Financial models driven by CBITCL processes are inherently incomplete and, therefore, there exist infinitely many risk-neutral measures beyond the probability measures considered in Corollary \ref{thm:girsanovcbitclcurrency}. Our approach is motivated by the preservation of the structure of the model under each risk-neutral measure $\QQ^i$, for $i=1,\ldots,N$. In line with the martingale approach to financial modeling, the parameters characterizing the family $\{\QQ^i; i=1,\ldots,N\}$ are determined by calibration to market data (see Section \ref{sec:numericalcbitclcurrency} for a specific application). We refer to \cite{EK20} for an overview of several well-known hedging approaches in incomplete markets driven by jump processes.
\end{rem}

By combining \eqref{eq:artificialfxrate} and \eqref{eq:fxrate}, we obtain the following representation of FX rates:
\be	\label{eq:fxrate_new}
S_t^{i,j} 
= S^{i,j}_0e^{(r^i-r^j)t}\prod_{k=1}^{d} e^{(\zeta^j_k-\zeta^i_k)X_t^k + (\lambda^j_k-\lambda^i_k)Z_t^k -(\cK_t^{j,k} - \cK_t^{i,k})}, 
\qquad \text{ for all } t \in[0,\cT].
\ee
In particular, note that all FX rates $S^{i,j}$ share the same modeling structure, for all $i,j=1,\ldots,N$, and the driving processes $(X^k,Z^k)$, for $k=1,\ldots,d$, remain mutually independent CBITCL processes under each risk-neutral measure $\QQ^i$, as a consequence of Corollary \ref{thm:girsanovcbitclcurrency}. 
In particular, the functional form of the process $S^{i,j}$ under $\QQ^i$ is identical to that of its reciprocal $S^{j,i}$ under $\QQ^j$.

We have thus proved the next theorem, which shows that we have constructed a well-posed and coherent multiple currency market, in line with the modeling objectives set in Section \ref{sec:fxmarket}.

\begin{thm}
The multiple currency market $(\bm{D},\bm{S})$ is well-posed and coherent.
\end{thm}

For each $i=1,\ldots,N$, the Radon-Nikodym density $\ud\QQ^i/\ud\QQ$ defined in \eqref{eq:measurechangecbitclcurrency} admits the following representation in terms of the sources of randomness driving the extended Dawson-Li representation \eqref{eq:cbitclsde1}-\eqref{eq:cbitclsde2} of the CBITCL processes $(X^k,Z^k)$, for $k=1,\ldots,d$:
\begin{align*}
\frac{\ud\QQ^i}{\ud\QQ} &= \prod_{k=1}^{d}\cE\left(\zeta^i_k\,\sigma^k\int_0^{\cdot} {\sqrt{X_s^k}\,\ud B_s^{k,1}} + \lambda^i_k\,\sigma_Z^k\int_0^{\cdot} {\sqrt{X_s^k}\,\ud B_s^{k,2}} +  \int_0^{\cdot}\int_0^{+\infty} {(e^{\zeta^i_kx}-1)\tildeN_0^k(\ud s,\ud x)} \right)_{\cT}\\
&\quad\times\prod_{k=1}^{d}\cE\left( \int_0^{\cdot}\int_0^{X_{s-}^k}\!\int_0^{+\infty} {(e^{\zeta^i_k x}-1)\tildeN_1^k(\ud s,\ud u,\ud x)} + \int_0^{\cdot}\int_0^{X_{s-}^k}\!\int_{\R} {(e^{\lambda^i_kx}-1)\tildeN_2^k(\ud s,\ud u,\ud x)} \right)_{\cT}.
\end{align*}
By Girsanov's theorem, the processes $B^{i,k,1}=(B^{i,k,1}_t)_{t\in[0,\cT]}$ and $B^{i,k,2}=(B^{i,k,2}_t)_{t\in[0,\cT]}$ defined by
\be\label{eq:girsanov_bm}\begin{aligned}
B_t^{i,k,1} &:= B_t^{k,1} - \zeta^i_k\,\sigma^k\int_0^t {\sqrt{X_s^k}\,\ud s},\\
B_t^{i,k,2} &:= B_t^{k,2} - \lambda^i_k\,\sigma_Z^k\int_0^t {\sqrt{X_s^k}\,\ud s},
\end{aligned}\ee
are independent Brownian motions under $\QQ^i$. Moreover, $N_0^k(\ud t, \ud x)$, $N_1^k(\ud t, \ud u, \ud x)$, $N_2^k(\ud t, \ud u, \ud x)$ are Poisson random measures under $\QQ^i$ with compensated measures
\be\label{eq:girsanov_rndmeas}\begin{aligned}
\tildeN_0^{i,k}(\ud t, \ud x) &:= N_0^k(\ud t, \ud x) - \ud t\,e^{\,\zeta^i_k\,x}\,\nu^k(\ud x),\\
\tildeN_1^{i,k}(\ud t, \ud u, \ud x) &:= N_1^k(\ud t, \ud u, \ud x) - \ud t\,\ud u\,e^{\,\zeta^i_k\,x}\,\pi^k(\ud x),\\	
\tildeN_2^{i,k}(\ud t, \ud u, \ud x) &:= N_2^k(\ud t, \ud u, \ud x) - \ud t\,\ud u\,e^{\,\lambda^i_k\,x}\,\gamma_Z^k(\ud x).
\end{aligned}\ee

\begin{lem}	\label{lem:sdefxrate}
For each $i,j=1,\ldots,N$, the process $S^{i,j}$ satisfies the following dynamics under $\QQ^i$:
\begin{align} 
\frac{\ud S^{i,j}_t}{S^{i,j}_{t-}} &= (r^i - r^j)\ud t + \sum_{k=1}^{d} \sqrt{X_t^k}\,\Bigl( \sigma^k(\zeta^j_k - \zeta^i_k)\,\ud B_t^{i,k,1} + \sigma_Z^k\,(\lambda^j_k - \lambda^i_k)\,\ud B_t^{i,k,2} \Bigr)\notag\\
&\quad+ \sum_{k=1}^{d} \int_0^{+\infty} {\bigl(e^{(\zeta^j_k-\zeta^i_k)x}-1\bigr)\tildeN_0^{i,k}(\ud t,\ud x)}
+ \sum_{k=1}^{d} \int_0^{X_{t-}^k}\!\int_0^{+\infty} {\bigl(e^{(\zeta^j_k-\zeta^i_k)x}-1\bigr)\tildeN_1^{i,k}(\ud t,\ud u,\ud x)}\label{eq:sdefxrate}\\
&\quad+ \sum_{k=1}^{d} \int_0^{X_{t-}^k}\!\int_{\R} {\bigl(e^{(\lambda^j_k-\lambda^i_k)x}-1\bigr)\tildeN_2^{i,k}(\ud t,\ud u,\ud x)},
\notag
\end{align}
with the processes $B^{i,k,1}$, $B^{i,k,2}$ and the random measures $\widetilde{N}^{i,k}_0$, $\widetilde{N}^{i,k}_1$, $\widetilde{N}^{i,k}_2$  defined in \eqref{eq:girsanov_bm}-\eqref{eq:girsanov_rndmeas}.
\end{lem}
\begin{proof}
The claim follows from \eqref{eq:fxrate} by applying It\^o's product rule together with the dynamics \eqref{eq:sdeartificialfxrate} of the artificial FX rates $S^{0,i}$ and $S^{0,j}$, making use of the notation introduced above.
\end{proof}

In particular, we can notice that the dynamics of $S^{i,j}$ under $\QQ^i$ are functionally symmetric with respect to the dynamics of $S^{j,i}$ under $\QQ^j$, for all $i,j=1,\ldots,N$. This is a further evidence of the coherence of our modeling framework, in the sense of \cite{Gno17}.

As can be seen from equation \eqref{eq:sdefxrate}, in our framework FX rates possess stochastic dynamics that can capture the most relevant risk characteristics of the FX market, as discussed in the Introduction. More specifically, we can remark the following features (for simplicity of presentation, in the following discussion we consider a one-dimensional CBI process $X$):
\begin{description}
\item[Stochastic volatility] 
for all FX rates, both the diffusive volatility and the jump volatility are stochastic and depend on the current level of the CBI process $X$. In particular, since $X$ is a self-exciting process, this induces volatility clustering effects in the FX rates.
\item[Jump risk]
all FX rates are affected by three different jump terms, corresponding to the three integrals appearing on the right-hand side of \eqref{eq:sdefxrate}: the first integral results from the immigration of the CBI process $X$, the second integral is related to the branching property of $X$ and the third integral is generated by the L\'evy process defining the process $Z$, with a jump intensity proportional to  $X$. In particular, the last two integrals are affected by the self-exciting property of $X$ and can generate jump clusters in the dynamics of FX rates. The parameters $\zeta^j-\zeta^i$ and $\lambda^j-\lambda^i$ control the magnitude of these effects.
\item[Stochastic dependence of FX rates] 
the quadratic covariation among different FX rates exhibits a rich stochastic structure, with the presence of common jumps and with a jump intensity which is related to the current level of the self-exciting CBI process $X$. 
\item[Stochastic skewness]
the quadratic covariation $[S^{i,j},X]$ also has a rich stochastic structure. In turn, this induces a stochastic instantaneous correlation between each FX rate and its stochastic volatility. As explained in \cite{CHJ09,DFG11}, the presence of stochastic correlation is responsible for the existence of stochastic variations in the skew of FX implied volatilites. 
\end{description}

\subsection{Currency option pricing}
\label{sec:currencypricing}

The modeling framework constructed in Section \ref{sec:constructionmodelcbitclcurrency} is coherent and, therefore, retains full analytical tractability under all risk-neutral measures considered in Corollary \ref{thm:girsanovcbitclcurrency}. In particular,  we can derive an explicit representation of the characteristic function of each FX rate. In the following, we denote by $\EE^i$ the expectation under $\QQ^i$, for each $i=1,\ldots,N$.

\begin{lem}\label{lem:characteristicfunctionfxrate}
For all $i,j=1,\ldots,N$, the characteristic function of the process $(\log S_t^{i,j})_{t\in[0,\cT]}$ under $\QQ^i$ is given by
\be	\label{eq:char_fun}
\EE^i\bigl[ e^{\im u\log S_t^{i,j} } \bigr] 
= e^{\im u(\log S_0^{i,j} +(r^i - r^j)t)}\prod_{k=1}^d e^{\im u(\Psi^k(\zeta^i_k) - \Psi^k(\zeta^j_k))t+\cU^{i,k}(t, u_1^k, u_2^k, u_3^k) + \cV^{i,k}(t, u_1^k, u_2^k, u_3^k)X_0^k},
\ee
for all $(u,t) \in \R\times[0,\horizon]$, where $(\cU^{i,k}(\cdot, u_1^k, u_2^k, u_3^k), \cV^{i,k}(\cdot, u_1^k, u_2^k, u_3^k))$ is the unique solution to system \eqref{eq:Riccati1}-\eqref{eq:Riccati2} associated to $(X^k, Z^k)$ under $\QQ^i$ with 
\begin{equation*}
u_1^k = \im u(\zeta^j_k - \zeta^i_k), \quad u_2^k = \im u\bigl(\Phi^k(\zeta^i_k) + \Xi_Z^k(\lambda^i_k) - \Phi^k(\zeta^j_k) - \Xi_Z^k(\lambda^j_k)\bigr) \quad \text{and} \quad u_3^k = \im u(\lambda^j_k - \lambda^i_k).
\end{equation*}	
\end{lem}
\begin{proof}
In view of \eqref{eq:fxrate_new} and \eqref{eq:exponentialcompensator}, we have that
\begin{align*}
\EE^i\bigl[ e^{\im u\log S_t^{i,j} } \bigr] &= e^{\im u(\log S_0^{i,j} + (r^i - r^j)t)}\prod_{k=1}^d e^{\im u(\Psi^k(\zeta^i_k) - \Psi^k(\zeta^j_k))t}\\
&\quad\times\prod_{k=1}^d\EE^i\left[e^{\im u(\zeta^j_k - \zeta^i_k)X_t^k + \im u(\Phi^k(\zeta^i_k) + \Xi_Z^k(\lambda^i_k) - \Phi^k(\zeta^j_k) - \Xi_Z^k(\lambda^j_k))Y_t^k + \im u(\lambda^j_k - \lambda^i_k)Z_t^k }\right],
\end{align*}
where we have used the independence of the CBITCL processes $(X^k,Z^k)$, for $k=1,\ldots,d$, under $\QQ^i$ (see Corollary \ref{thm:girsanovcbitclcurrency}) and $Y_t^k = \int_0^t {X_s^k\,\ud s}$, for all $k=1,\ldots,d$. Formula \eqref{eq:char_fun} then follows from an application of the affine transform formula given in Proposition \ref{prop:cbitclaffine}.
\end{proof}

The availability of an explicit description of the characteristic function of each FX rate allows for currency option pricing via Fourier techniques. We adopt the COS method of \cite{FOo09}, which presents the advantage of utilizing only the characteristic function of the process, without requiring any domain extensions as in other Fourier pricing methods. In our setting, such domain extensions would necessitate additional constraints on the parameters, potentially affecting the calibration results. 
We consider a European Call option in the $i\textsuperscript{th}$ economy written on the FX rate $S^{i,j}$, with maturity $T \leq \horizon$ and strike $K > 0$. 
We assume that the distribution of $\log S^{i,j}_T$ under $\QQ^i$ admits a density\footnote{By \cite[Theorem 16.6]{Williams}, the random variable $\log S^{i,j}_T$ admits a density under $\QQ^i$ if $\int_{\R}|\EE^i[e^{\im u\log S_T^{i,j} }]|\ud u<+\infty$. Under this assumption, the density can be recovered by Fourier inversion from the characteristic function of $\log S^{i,j}_T$ given in Lemma \ref{lem:characteristicfunctionfxrate}.}.
Since the multiple currency market is well-posed, we can apply risk-neutral valuation under $\QQ^i$ to compute the arbitrage-free price $C(T,K)$ at $t=0$ of the option:
\be	\label{eq:COS_integral}
C(T, K) = e^{-r^iT}\,\EE^i\bigl[(S_T^{i,j} - K)^{+}\bigr] = e^{-r^iT}\int_{\R} {K(e^x - 1)^{+} f_T^{i,j}(x)\,\ud x},
\ee
where $f_T^{i,j}$ represents the density function of $\log(S_T^{i,j}/K)$ under $\QQ^i$.
To compute the integral in \eqref{eq:COS_integral}, we introduce a suitably chosen truncation range $[a, b] \subset \R$ such that $C(T, K)$ can be approximated with good accuracy by
\[
C(T, K) \approx e^{-r^i\,T}\int_a^b {K(e^x - 1)^{+} f_T^{i,j}(x)\,\ud x}.
\]
The resulting pricing formula is stated in the next proposition, which follows by the same arguments presented in \cite[Section 2.1]{FOo09}. 

\begin{prop}\label{prop:currencypricing}
The arbitrage-free price $C(T, K)$ of a European call option written on the FX rate $S^{i,j}$, with maturity $T \leq \horizon$ and strike $K > 0$, can be approximated by
\be\label{eq:currencypricing}
C(T, K) \approx e^{-r^iT}\,K\sum_{k=0}^{M-1}\left(1 - \frac{\delta_0(k)}{2}\right)\Re\left(e^{\im\frac{k\pi}{a-b}\left(a + \log K\right)}\EE^i\Bigl[e^{\im\frac{k\pi}{b-a}\log S_T^{i,j}}\Bigr]\right)B_k,
\ee
where $\delta_0$ denotes the Kronecker delta at $0$, $M \in \N$, $B_0 = (e^{\,b} - 1 - b)/(b-1)$, and where 
\[
B_k = \frac{2}{b-a}\left(\frac{1}{1 + \left(\frac{k\,\pi}{b-a}\right)^2}\left((-1)^k\,e^{\,b} - \cos\left(\frac{k\,\pi\,a}{b-a}\right) + \frac{k\,\pi}{b-a}\sin\left(\frac{k\,\pi\,a}{b-a}\right)\right) - \frac{b-a}{k\,\pi}\sin\left(\frac{k\,\pi\,a}{b-a}\right)\right).
\]
\end{prop}

\begin{rem}\label{rem:truncationrange}
(1) In order to ensure the accuracy of formula \eqref{eq:currencypricing}, one needs to specify the truncation range $[a, b]$ properly. Following \cite[Section 5.1]{FOo09}, a suitable specification is the following:
\[
[a, b] = \left[c_1 - L\,\sqrt{c_2 + \sqrt{c_4}}, \quad c_1 + L\,\sqrt{c_2 + \sqrt{c_4}}\right],
\] 
with $L = 10$ and where $c_n$, for $n = 1, 2, 4$, represents the $n\textsuperscript{th}$ cumulant of $\log(S_T^{i,j}/K)$. In our framework, the cumulants are not available in closed form. However, they can be approximated by using finite differences since, by definition, they are given by the derivatives at zero of the cumulant-generating function of $\log(S_T^{i,j}/K)$ (see \cite[Appendix A]{FOo09} for further details).

(2) As explained in \cite[Section 3.3]{FOo09}, formula \eqref{eq:currencypricing} can be readily extended to a multi-strike setting, which is practically important when one needs to price several options with the same maturity but associated to different strikes (e.g., for model calibration).
We refer to \cite[Remark 5.13]{phdthesis_szulda} for a description of the multi-strike implementation of formula \eqref{eq:currencypricing}.
\end{rem}

\section{Model calibration}\label{sec:numericalcbitclcurrency}

In this section, we calibrate a simple specification of our modeling framework to market data on a currency triangle. We consider a model driven by  tempered $\alpha$-stable CBI processes and CGMY L\'evy processes (see \cite{FGS21}) and propose two different calibration methods, one based on standard techniques and one relying on a deep learning algorithm. The market data are described in Section \ref{sec:fxmarketdata}, while the two calibration methods are presented in Section \ref{sec:calibrationtwotypes}. Section \ref{sec:specification} contains a description of the model specification and in Section \ref{sec:fxcalibrationresults} we report the calibration results.

\subsection{FX market data}\label{sec:fxmarketdata}

We consider market data on three FX implied volatility surfaces: EUR-USD, EUR-JPY and USD-JPY (according to the FOR-DOM convention, the second currency of each pair represents the domestic currency). The quoting convention for FX implied volatilities differs from the case of equity markets, since implied volatilities are quoted in terms of deltas and maturities instead of strikes and maturities. Moreover, excluding ATM options, individual volatilities are not directly quoted: the market practice consists in quoting certain combinations of contracts (risk-reversals and butterflies) from which implied volatilities for single contracts in terms of maturities and deltas have to be recovered.

For the three volatility surfaces, we consider a common set of maturities, ranging from one week to one year (1, 2 weeks, 1, 3, 6 months, and 1 year, representing the most liquid part of the implied volatility surface). We retrieved from Bloomberg the following market quotes as of April 15, 2020: ATM implied volatility, $10\Delta$ and $25\Delta$ risk-reversals\footnote{By $25\Delta$ risk-reversal, we mean an OTM Call option with a delta of $25\%$ and a Put option with a delta of $-25\%$.} and butterflies. For $25\Delta$, we have
\[
RR_{25\Delta} = \sigma_{25\Delta Call} - \sigma_{25\Delta Put},
\qquad\text{ and }\qquad
BF_{25\Delta} = \frac{\sigma_{25\Delta Call} + \sigma_{25\Delta Put}}{2} - \sigma_{A T M},
\]
from which we deduce 
\[
\sigma_{25\Delta Call} = \sigma_{A T M}+\frac{1}{2}\,RR_{25\Delta}+BF_{25\Delta}
\qquad\text{ and }\qquad
\sigma_{25\Delta Put} = \sigma_{A T M}-\frac{1}{2}\,RR_{25\Delta}+BF_{25\Delta},
\]
and similarly for $10\Delta$. For each currency pair and for each maturity, we have the implied volatilities of 5 contracts at our disposal. Market data not corresponding to the 5 points is interpolated. 

In order to reconstruct observed market prices, we also retrieved from Bloomberg FX spots and FX forward points, which enable us to build FX forward curves by adding the spot and the forward points. Equipped with such data, we have all the information needed to convert deltas into strikes and recover implied volatilities for single contracts in terms of maturities and strikes.\footnote{We performed these tasks by using the open-source Java library Strata by OpenGamma, available at \url{https://github.com/OpenGamma/Strata}.}

\subsection{Two calibration methods}\label{sec:calibrationtwotypes}

Let $p$ denote a vector of model parameters, belonging to some set of admissible parameters $\cP$. Let $\# T$ be the number of maturities and $\# K$ be the number of strikes that we consider. For simplicity of presentation, we assume that all volatility surfaces have the same strike range and the same number of strikes. 
In general, a calibration to the implied volatilites on a set of $N$ currencies consists in solving the following minimization problem:
\be\label{eq:fxcalibrationproblem}
\min_{p\in\cP}\sum_{u=1}^N\sum_{i = 1}^{\# T}\sum_{j =1}^{\# K}\left(\sigma_{imp}^{mkt}(u,T_i,K_j)-\sigma_{imp}^{mod(p)}(u,T_i,K_j)\right)^2,
\ee
where $\sigma_{imp}^{mkt}(u,T_i,K_j)$ denotes the implied volatility observed on the market for currency $u$, maturity $T_i$, and strike $K_j$, while $\sigma_{imp}^{mod(p)}(u,T_i,K_j)$ denotes its model-implied counterpart for a given vector of parameters $p \in \cP$. 

We now present two calibration methods. The first one, to which we refer as \emph{standard calibration}, utilizes pricing formula \eqref{eq:currencypricing} to compute model prices for a given choice of model parameters. Such prices are then converted into model-implied volatilities and inserted into \eqref{eq:fxcalibrationproblem}. This gives rise to a multi-dimensional function $\Sigma:\cP\to\R^{N\times\# T\times\# K}$ such that, for all $p\in\cP$, we have $\Sigma(p)_{(u,i,j)} = \sigma_{imp}^{mod(p)}(u,T_i,K_j)$, for every $u=1,\ldots,N$, $i=1,\ldots,\# T$, and $j=1,\ldots,\# K$.

The second calibration method, to which we refer as \emph{deep calibration}, adopts the two-step approach developed by \cite{HMT21} for the solution of \eqref{eq:fxcalibrationproblem}. We proceed as follows.
\begin{description}
\item[Grid-based implicit training] the purpose of this step is to approximate the non-linear function $\Sigma$ by a fully-connected feed-forward neural network $\cN^w:\cP\to\R^{N\times\# T\times\# K}$ (see \cite[Definition 1]{HMT21}), where $w$ denotes a vector of network parameters (typically weights and biases). We divide this step into two sub-steps:
\begin{enumerate}
\item We generate a training set $\{(p_n, \Sigma(p_n))\}_{n=1,\ldots,N_{train}}$ of size $N_{train}$, where each vector of parameters $p_n$ is generated randomly by means of a standard random generator (suitable adjustments can be made to guarantee that parameter restrictions are satisfied), and where we have fixed the grid $(u, T_i, K_j)$, $u=1,\ldots,N$, $i=1,\ldots,\# T$, and $j=1,\ldots,\# K$, throughout the generation (hence the term ``grid-based'').
\item We solve the following minimization problem called ``training'' of the neural network:
\be\label{eq:trainingproblem}
\min_{w}\sum_{n=1}^{N_{train}}\sum_{u=1}^N\sum_{i = 1}^{\# T}\sum_{j =1}^{\# K}\Bigl(\Sigma(p_n)_{(u,i,j)}-\cN^w(p_n)_{(u,i,j)}\Bigr)^2,
\ee
whose solution is an optimal vector of network parameters $\widehat{w}$ such that the neural network $\cN := \cN^{\widehat{w}}$ best approximates the observations $\{\Sigma(p_n)\}_{n=1,\ldots,  N_{train}}$. Notice that $\widehat{w}$ depends on the grid that we have fixed, thus explaining the term ``implicit''.
\end{enumerate}
\item[Deterministic calibration] we rewrite \eqref{eq:fxcalibrationproblem} with the trained neural network $\cN$ as follows:  
\be\label{eq:fxdeepcalibrationproblem}
\min_{p\in\cP}\sum_{u=1}^N\sum_{i = 1}^{\# T}\sum_{j =1}^{\# K}\left(\sigma_{imp}^{mkt}(u,T_i,K_j)-\cN(p)_{(u,i,j)}\right)^2.
\ee
\end{description}
Following \cite{HMT21}, we adopt the following neural network architecture:
\begin{itemize}
\item 3 hidden layers with 30 nodes on each;
\item $N =3$ surfaces, all sharing the same maturity range of size $\# T = 6$ and the same number of strikes $\# K = 5$. This yields an output layer of $3 \times 6 \times 5 = 90$ nodes. The size of the input layer is simply the number of model parameters;
\item On the input and hidden layers, we employ the Exponential Linear Unit (ELU) activation function. The output layer is in turn equipped with the Sigmoid function.
\end{itemize}

Figure \ref{fig:neuralnet} provides a visualization of the neural network architecture. 
	\begin{figure}[h!]
		\centering
		\includegraphics[scale=0.2]{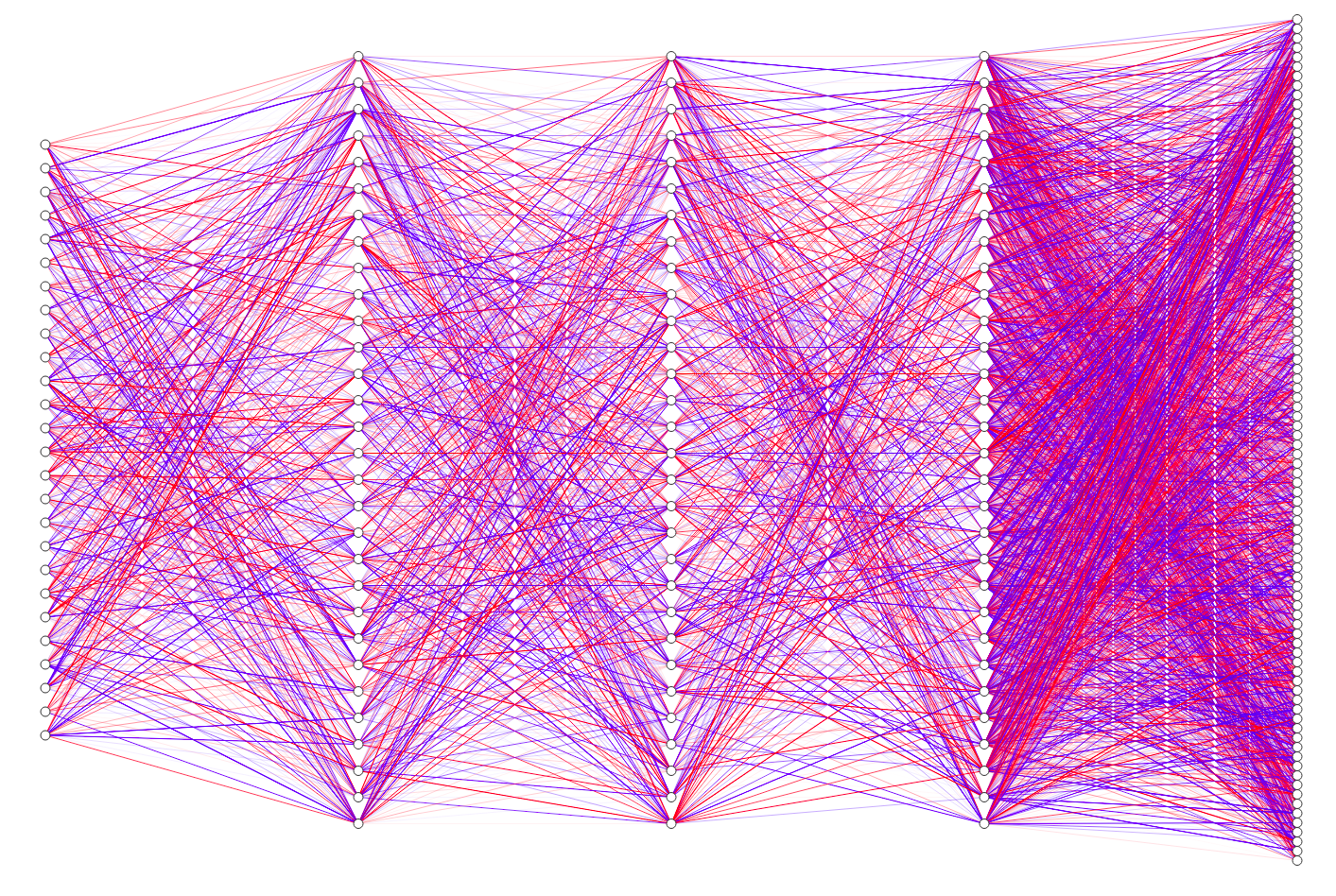}
		\caption{Illustration of the chosen neural network architecture, where all weights have been generated randomly. The width of an edge is proportional to its weight. The color of an edge defines the sign of the weight (red if positive, blue if negative). The figure has been generated using the method of \cite{Lenail19}.}
		\label{fig:neuralnet}
	\end{figure}
	
For the training step, we start with the random generation of a training set of size $N_{train} = 10.000$. After normalization, we proceed with the training of the neural network by solving \eqref{eq:trainingproblem}. The common practice is to use a stochastic optimization algorithm based on ``mini-batch'' gradient descent (see \cite{GBC16}), whose updater can be specified following the \emph{Adam} scheme (see \cite{KB17}).  We set the mini-batch size to 32 and the number of epochs to 150 with potential early stopping\footnote{We rely on the open-source Java library Eclipse Deeplearning4j, available at \url{http://deeplearning4j.org}.}.
	
\subsection{CBITCL specification}	\label{sec:specification}

We consider a specification of the modeling framework described in Section \ref{sec:constructionmodelcbitclcurrency} driven by two independent CBITCL processes $(X^k,Z^k)$, $k=1,2$, where
\begin{enumerate}
\item[(i)] $X^1$ and $X^2$ are tempered $\alpha$-stable CBI processes, as defined in \cite{FGS21};
\item[(ii)] the L\'evy processes $L^1$ and $L^2$ generating the processes $Z^1$ and $Z^2$ (see Definition \ref{def:cbitcl}), respectively, are CGMY processes, as introduced in \cite{CGMY02}.
\end{enumerate}

We recall from \cite{FGS21} that a CBI process $X=(X_t)_{t\in[0,\cT]}$ is said to be {\em tempered $\alpha$-stable} if its immigration mechanism  \eqref{eq:immigrationmechanism} reduces to $\Psi(x)=\beta x$ and the measure $\pi$ in the branching mechanism \eqref{eq:branchingmechanism} corresponds to the L\'evy measure of a spectrally positive tempered $\alpha$-stable compensated L\'evy process. More specifically, we set
\[
\pi(\ud z) = C_{\alpha}\,\eta^{\alpha}\frac{e^{-\frac{\theta}{\eta}z}}{z^{1+\alpha}}\ind_{\{z>0\}}\ud z,
\]
where $\eta>0$, $\theta\geq0$, $\alpha\in(-\infty,2)$ (restricted to $\alpha\in(1,2)$ if $\theta=0$) and $C_{\alpha}$ is a normalization constant.
This family of processes represents the tempered version of $\alpha$-stable CBI processes, which have been successfully applied in finance in recent years (see, e.g., \cite{JMS17,JMSS19,JMSZ21}).
The parameter $\alpha$ is referred to as the {\em stability index} and determines the jump behavior of the process $X$ (see \cite[Section 3.1]{FGS21}):
\begin{itemize}
\item if $\alpha<0$, then $X$ has jumps of finite activity and finite variation;
\item if $\alpha\in[0,1)$, then $X$ has jumps of infinite activity and finite variation;
\item if $\alpha\in[1,2)$, then $X$ has jumps of infinite activity and infinite variation.
\end{itemize}
In the present setting, we shall consider the case $\alpha\in(1,2)$ and specify the normalization constant as $C_{\alpha}=1/\Gamma(-\alpha)$.
The jumps of the process $X$ are tempered exponentially depending on the value of the parameter $\theta$, while the parameter $\eta$ serves as a volatility coefficient controlling the jump volatility of the process $X$.
The following lemma provides the explicit representation of the branching mechanism $\Phi$ of a tempered $\alpha$-stable CBI process $X$ (we refer to \cite{phdthesis_szulda} for a proof). We denote by $\Gamma$ the Gamma function extended to $\R\setminus\mathbb{Z}_-$ (see \cite{Lebedev}).

\begin{lem}\label{lem:temperedstablecbi}
For a tempered $\alpha$-stable CBI process $X$ with $\eta > 0$, $\theta \geq 0$, $C_{\alpha}=1/\Gamma(-\alpha)$ and $\alpha \in (1,2)$, the set $\cD_X$ defined in \eqref{eq:domaincbi} is given by $\cD_X=(-\infty,\theta/\eta]$. Moreover, the branching mechanism $\Phi$ is given by
\[
\Phi(x) = -\,b\,x + \frac{1}{2}\,(\sigma\,x)^2 + (\theta - \eta\,x)^{\alpha} - \theta^{\alpha} + \alpha\,\theta^{\alpha- 1}\,\eta\,x,
\qquad\text{ for all }x \leq \theta/\eta.
\]
\end{lem} 
 
Let us also recall from \cite{CGMY02} that a L\'evy process $L=(L_t)_{t\in[0,\cT]}$ is of CGMY type if its L\'evy measure $\gamma$ is given by
\[
\gamma(\ud z) = C_L\bigl(z^{-1-Y}\,e^{-\,M\,z}\,\ind_{\{z>0\}} + |z|^{-1-Y}\,e^{-\,G\,|z|}\,\ind_{\{z<0\}}\bigr)\ud z, 
\]
where we fix $C_L=1/\Gamma(-Y)$.
The parameter $G > 0$ tempers the downward jumps of $L$, while $M > 0$ tempers the upward jumps, and $Y \in (1,2)$ controls the local behavior of $L$ in a similar way to the parameter $\alpha$ above.
We recall that the L\'evy exponent $\Xi$ of a CGMY process is of the form
\be\label{eq:cbitclcurrencyinfiniteactivitylevykhintchine}
\Xi(u) := \beta\,u + \int_{\R} {(e^{zu} - 1 - zu)\gamma(\ud z)}, \qquad\text{ for all } u \in \im\R,
\ee
for $\beta\in\R$.
It can be easily checked that, in the case of a CGMY process, the set $\cD_Z$ defined in \eqref{eq:domaincbitcl} is given by $\cD_Z=[-G,M]$. The L\'evy exponent \eqref{eq:cbitclcurrencyinfiniteactivitylevykhintchine} then takes the following explicit form:
\[
\Xi(u) = \beta u + (M - u)^{Y} - M^{Y} + (G + u)^{Y} - G^{Y} + u\,Y\,(M^{Y - 1} - G^{Y - 1}),
\quad \text{ for all }u\in[-G,M].
\end{equation*}

As discussed in Section \ref{sec:currencypricing}, the pricing of currency option requires the transformation of the model under $\QQ^i$, for $i=1,2,3$, where each $\QQ^i$ represents the risk-neutral measure associated to the $i\textsuperscript{th}$ economy and is given by \eqref{eq:measurechangecbitclcurrency}. In the present model specification, Theorem  \ref{thm:girsanovcbitclcurrency} directly implies the following result, which shows that not only the general CBITCL structure, but also the tempered $\alpha$-stable property of $X^k$ and the CGMY structure of $L^k$, for $k=1,2$, is preserved.

\begin{cor}
Under the model specification considered in this section, let  $\QQ^i$ be the probability measure defined in \eqref{eq:measurechangecbitclcurrency}, for each $i=1,\ldots,N$. Then,  the processes $(X^k,Z^k)$, $k=1,2$, remain independent CBITCL processes under $\QQ^i$ and such that
\begin{enumerate}
\item[(i)] $X^k$ is a tempered $\alpha$-stable CBI process with tempering parameter $\theta^{i,k}=\theta^k-\zeta^i_k\eta^k$;
\item[(ii)] $Z^k=L^k_{Y^k}$, where $L^k$ is a CGMY process with tempering parameters $G^{i,k}=G^k+\lambda^i_k$ and $M^{i,k}=M^k-\lambda^i_k$.
\end{enumerate}
\end{cor}

Moreover, the drift term $b^{i,k}$ and $b_Z^{i,k}$ given in Table \ref{table:newparameterscbitclcurrency} can be explicitly computed as follows (see \cite[Chapter 2]{phdthesis_szulda} for further details), for $i=1,\ldots,N$ and $k=1,2$:
\begin{align*}
b^{i,k} &= b_k - \zeta^i_k\,\sigma_k^2 - \alpha_k\,\eta_k^{\alpha_k}\Bigl(\theta_k^{\alpha_k - 1} - (\theta_k - \zeta^i_k\,\eta_k)^{\alpha_k - 1}\Bigr),\\
\beta_Z^{i,k} &= \beta_Z^k + Y^k\,\Bigl( (M^k)^{Y^k - 1} - (M^k - \lambda^i_k)^{Y^k - 1} + (G^k)^{Y^k - 1} - (G^k + \lambda^i_k)^{Y^k - 1} \Bigr).
\end{align*}

\subsection{Calibration results}\label{sec:fxcalibrationresults}

For the resolution of \eqref{eq:fxcalibrationproblem} and \eqref{eq:fxdeepcalibrationproblem}, we use the Levenberg--Marquardt optimizer of the open-source Java library Finmath\footnote{Available at \url{https://www.finmath.net/finmath-lib}.}. 
We perform standard and deep calibrations. For the standard one, we obtain a root-mean-square error of 0.07557 in 709.977 seconds. Figure \ref{fig:cgmystandard} shows a satisfactory fit that slightly worsens for longer maturities. The deep calibration outperforms the standard one, achieving a root-mean-square error of 0.04092 in 0.269 seconds. The better quality of the fit can be seeen from Figure \ref{fig:cgmydeep}, where we can observe an improvement for longer maturities. Moreover, the execution time of the deep calibration is much smaller. However, one should take into account the time required for the training step, which may last up to several hours.
 
The calibrated values of the parameters are reported in Table \ref{table:fxcalibratedparameters}. In particular, we can notice that the calibrated values of $\alpha_1$ and $\alpha_2$ are rather close to 1, indicating the potential presence of jump clustering phenomena (see \cite[Section 3.1]{FGS21} for a detailed discussion of this aspect and for an analogous empirical evidence in multi-curve interest rate markets).
We can also observe that the differences $\zeta_{\uEUR}-\zeta_{\uUSD}$, $\zeta_{\uEUR}-\zeta_{\uJPY}$, $\zeta_{\uUSD}-\zeta_{\uJPY}$ show evidence of moderate dependence between the FX rates considered and their volatility, in line with the findings of \cite{BM18}. By inspecting the calibrated values of the parameters of the tempered $\alpha$-stable CBI processes, we can also notice a non-trivial contribution from the self-exciting jumps.
	 
It is interesting to remark that the calibrated values of the parameters are stable across the two types of calibration. This is in accordance with the order of the two calibration exercises: first, we performed the deep calibration, where the initial guess was generated randomly by employing the same random generator that we used for the generation of the training set. We then used the optimal set of parameters obtained from this first calibration as the initial guess of the standard calibration. The fact that the output of the standard calibration is in line with that of the deep calibration provides us with a way of validating, in the present context, the deep calibration.
	
In order to assess the importance of allowing for jumps in the CBITCL specification, we compared the specification described in Section \ref{sec:specification} with a continuous-path model where the L\'evy processes $L^1$ and $L^2$ are simply given by two independent Brownian motions and the CBI processes $X^1$ and $X^2$ are standard square-root diffusions. This simplified model results in a Heston-type model. For the comparison, we calibrated both models to the same set of market implied volatilites, employing the standard calibration technique described in Section \ref{sec:calibrationtwotypes}. Due to the much simpler structure of the model (in particular, of the associated Riccati equations), the calibration of the Heston-type model requires 311 seconds, while the calibration of the CBITCL model described in Section \ref{sec:specification} required 703 seconds. However, the Heston-type model exhibits a worse fit to market data, achieving a RMSE of 0.1236. We observe that, in spite of the exclusion of the jump components, the calibrated values of the remaining parameters are quite similar across the two different specifications. These findings suggest that the jump components of our CBITCL specification capture some features of market data that cannot be adequately reproduced by continuous-path models.

\begin{table}[t]
		\centering
		\begin{tabular}{|c|c|c||c|c|c|}
			\hline 
			& Standard & Deep & & Standard & Deep \\
			\hline\hline
			$X^1_0$ & $1.1040$ & $1.1106$ & $X^2_0$ & $0.19652$ & $0.18549$ \\
			$\beta^1$ & $0.37721$ & $0.65766$ & $\beta^2$ & $1.7524$ & $1.7782$ \\
			$b^1$  & $0.43082$ & $0.43082$ & $b^2$ & $-0.73467$ & $-0.73467$ \\
			$\sigma^1$ & $2.1473$ & $2.1473$ & $\sigma^2$ & $1.1174$ & $1.1174$ \\
			$\eta^1$ & $1.7208$ & $1.7208$ & $\eta^2$ & $2.1855$ & $2.1855$ \\
			$\theta^1$ & $1.9338$ & $1.9338$ & $\theta^2$ & $0.65273$ & $0.65273$ \\
			$\alpha^1$ & $1.1697$ & $1.1697$ & $\alpha^2$ & $1.1122$ & $1.1122$ \\
			$\beta^1_Z$ & $-0.16220$ & $-0.16220$ & $\beta^2_Z$ &$0.88065$ & $0.88065$ \\
			$G^1$ & $3.0313$ & $3.0313$ & $G^2$ & $0.59711$ & $0.59711$ \\
			$M^1$ & $0.79529$ & $0.79529$ & $M^2$ & $0.22821$ & $0.22821$ \\
			$Y^1$ & $1.7675$ & $1.7675$ & $Y^2$ & $1.2390$ & $1.2390$ \\
			$\zeta_{JPY,1}$ & $1.12323$ & $1.12366$ & $\zeta_{JPY,2}$ & $0.232636$ & $0.232636$ \\
			$\zeta_{USD,1}$ & $0.27244$ & $0.27244$ & $\zeta_{USD,2}$ & $0.092184$ & $0.060470$ \\
			$\zeta_{EUR,1}$ & $0.089747$ & $0.097352$ & $\zeta_{EUR,2}$ & $0.025973$ & $0.024422$ \\
			$\lambda_{JPY,1}$ & $0.39764$ & $0.39764$ & $\lambda_{JPY,2}$ & $0.11410$ & $0.11410$ \\
			$\lambda_{USD,1}$ & $0.32863$ & $0.32863$ & $\lambda_{USD,2}$ & $-0.014839$ & $-0.014839$ \\
			$\lambda_{EUR,1}$ & $0.16260$ & $0.16260$ & $\lambda_{EUR,2}$ & $0.040496$ & $0.040496$ \\
			\hline
		\end{tabular} 
		\caption{Calibrated values of the model parameters. }
		\label{table:fxcalibratedparameters}
	\end{table}
	
	\begin{figure}[p]
		\begin{subfigure}{\textwidth}
			\centering
			\includegraphics[scale=0.45]{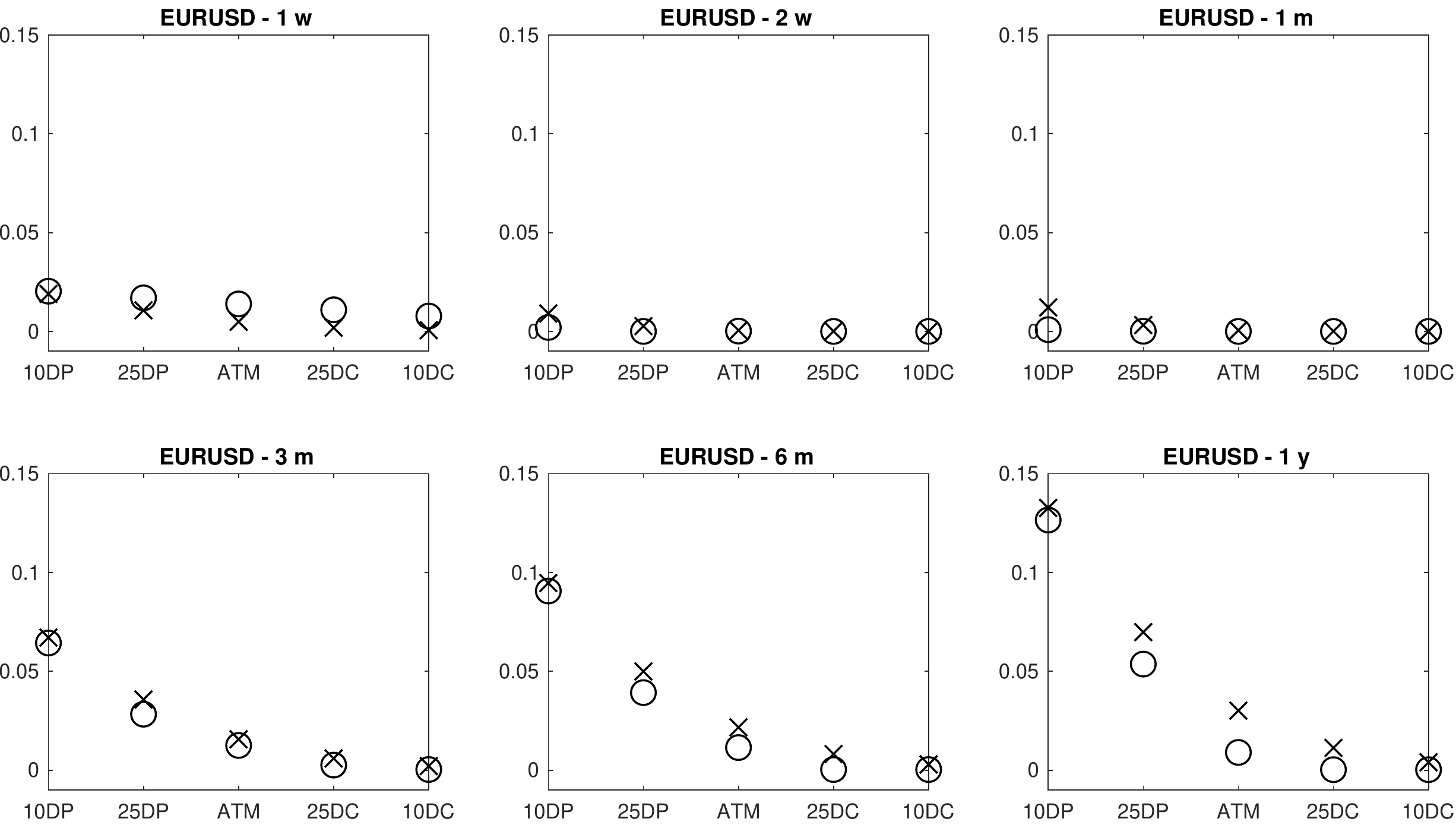}
		\end{subfigure}
		\begin{subfigure}{\textwidth}
			\centering
			\includegraphics[scale=0.45]{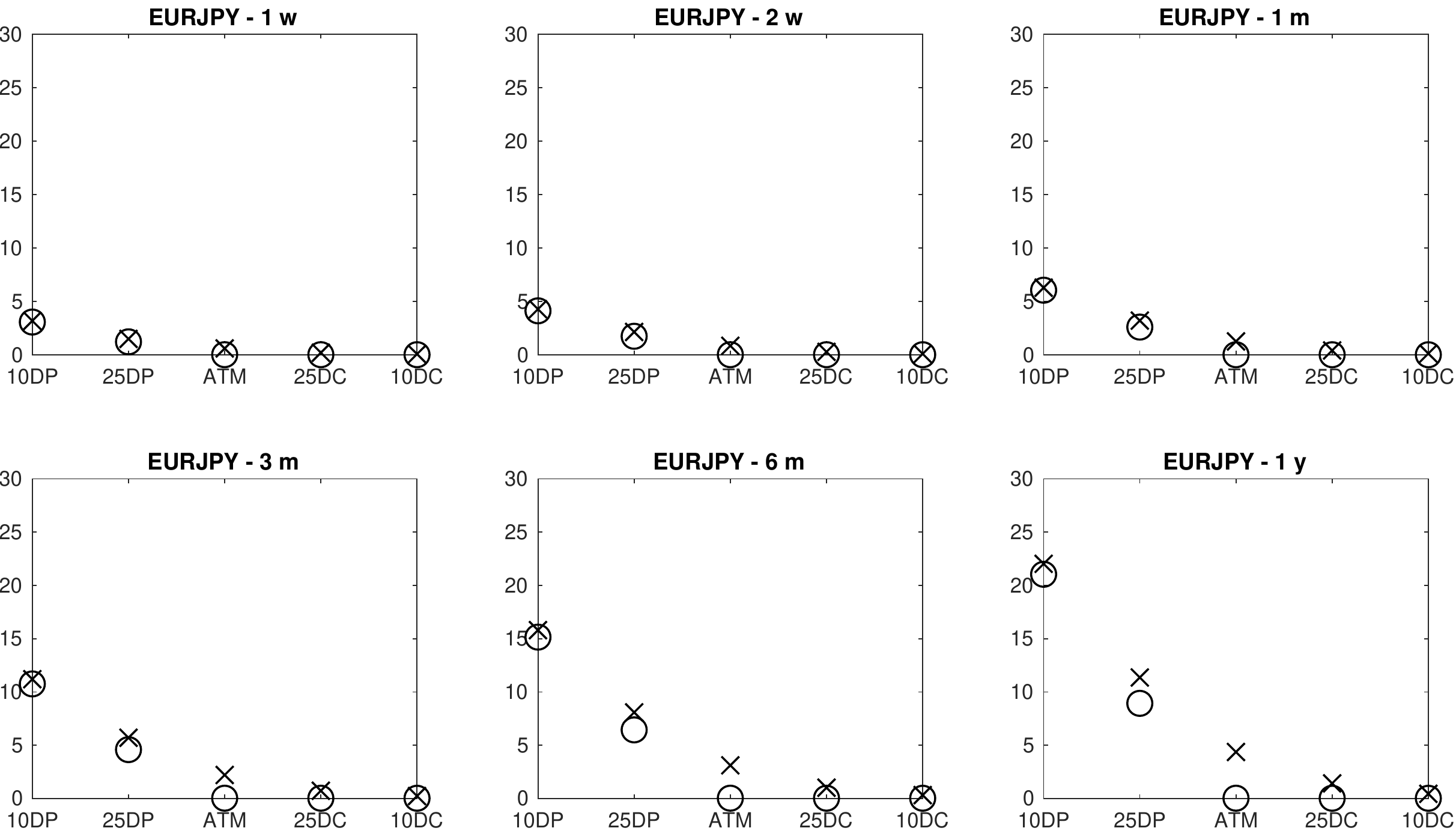}
		\end{subfigure}
		\begin{subfigure}{\textwidth}
			\centering
			\includegraphics[scale=0.45]{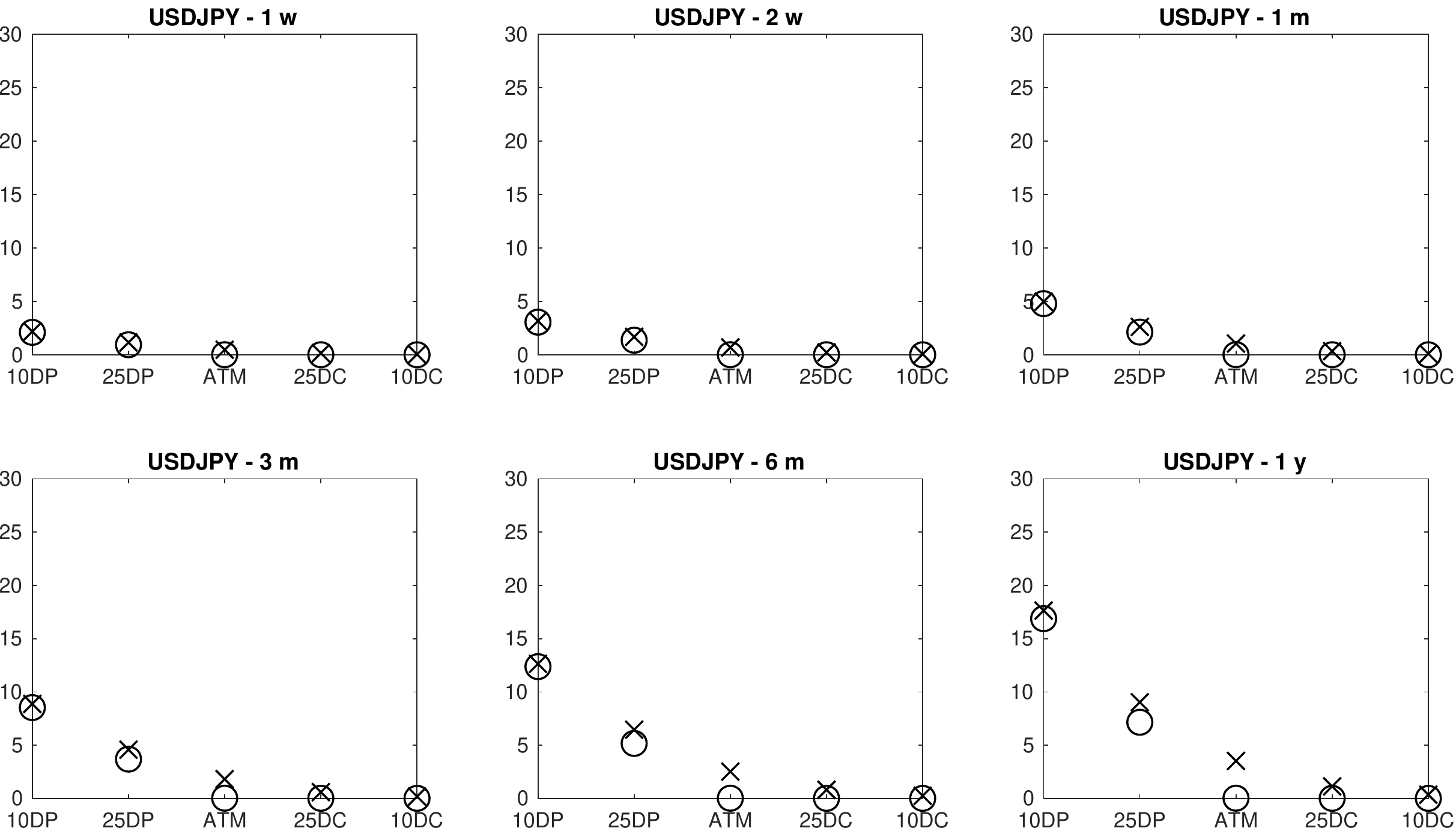}
		\end{subfigure}
		\caption{Calibration results obtained via the standard calibration. Market prices are denoted by crosses, model prices are denoted by circles. Moneyness levels follow the standard Delta quoting convention in the FX option market. DC and DP stand for ``delta call'' and ``delta put'', respectively.}
		\label{fig:cgmystandard}
	\end{figure}
	
	\begin{figure}[p]
		\begin{subfigure}{\textwidth}
			\centering
			\includegraphics[scale=0.45]{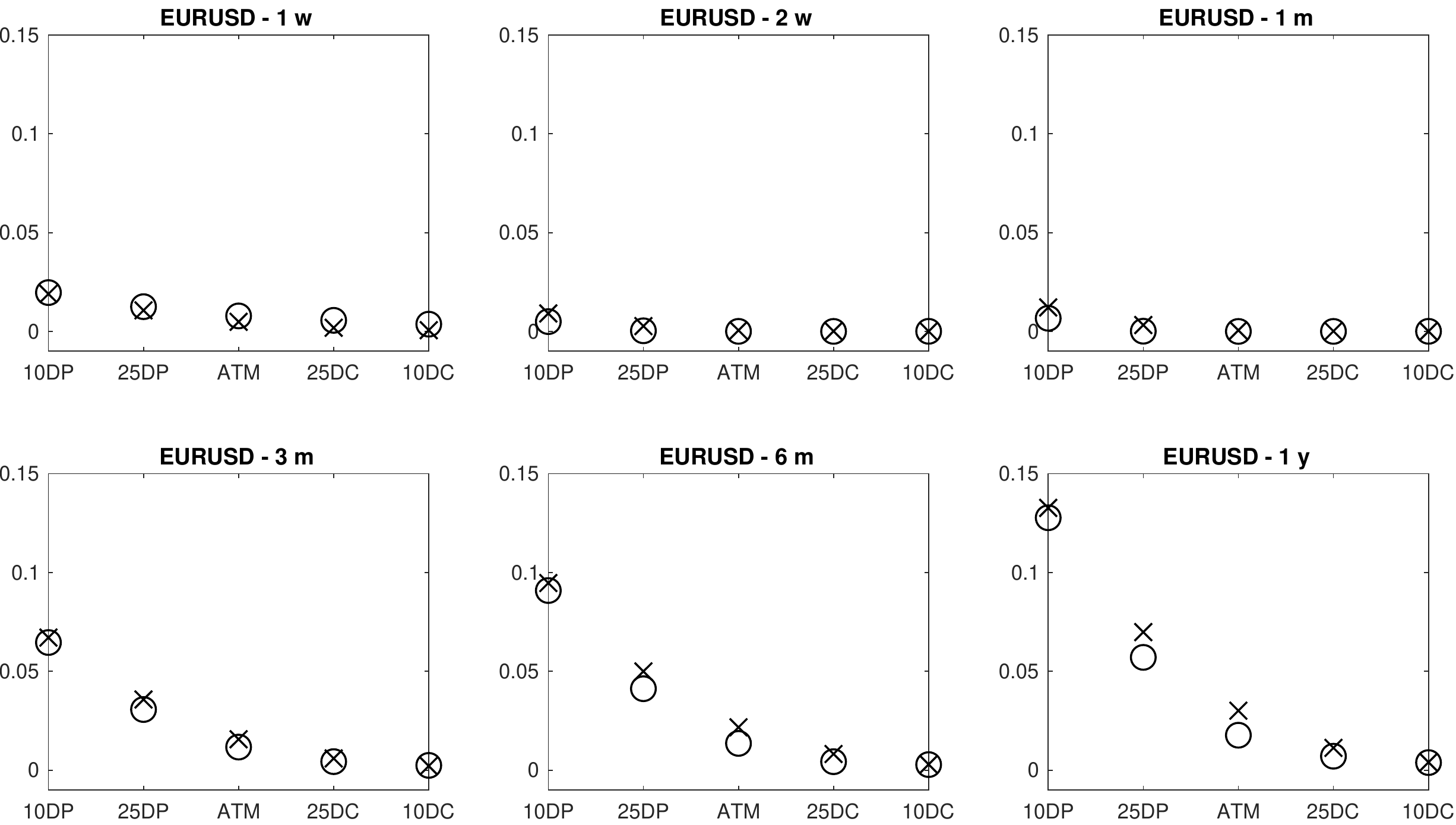}
		\end{subfigure}
		\begin{subfigure}{\textwidth}
			\centering
			\includegraphics[scale=0.45]{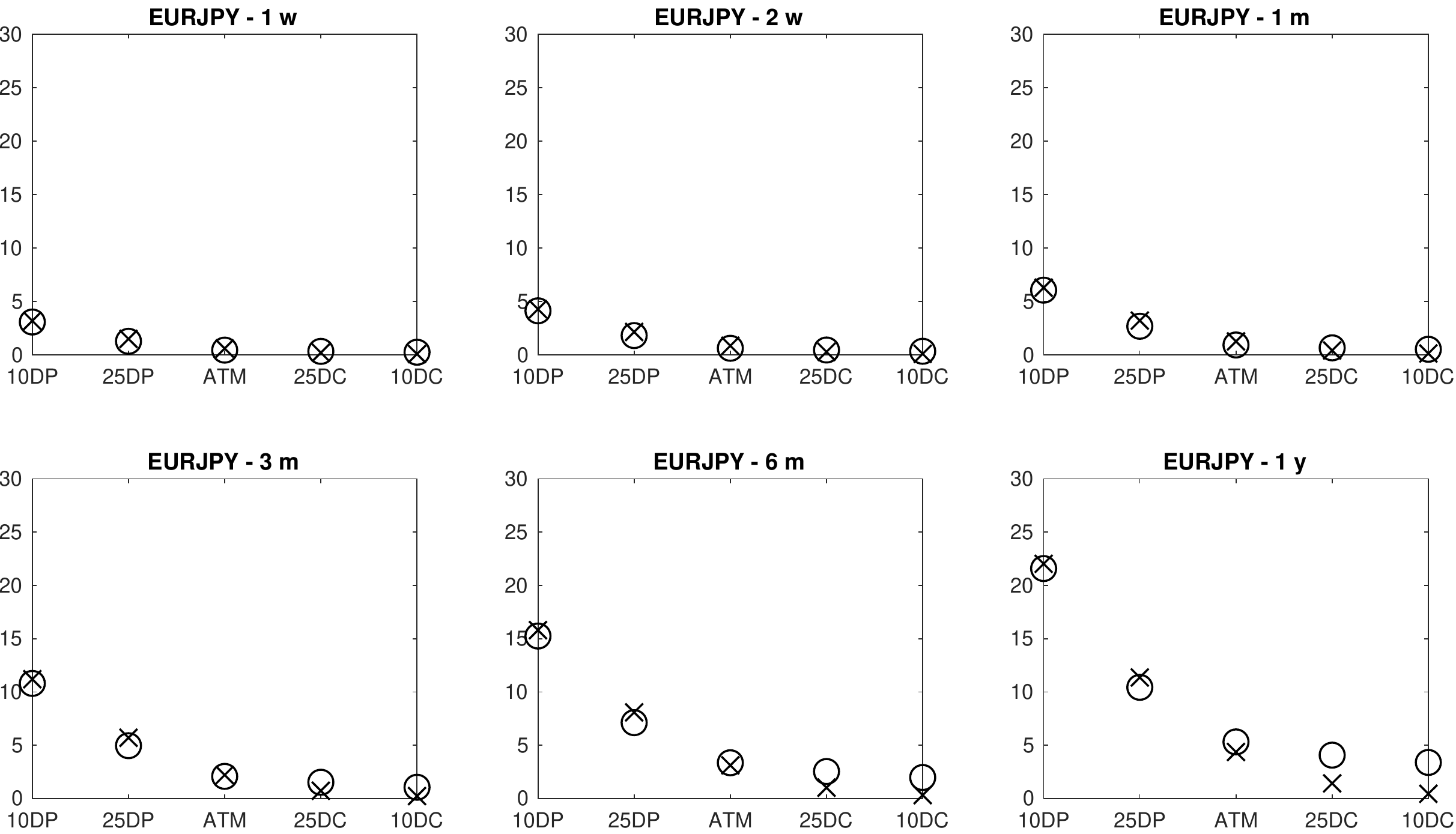}
		\end{subfigure}
		\begin{subfigure}{\textwidth}
			\centering
			\includegraphics[scale=0.45]{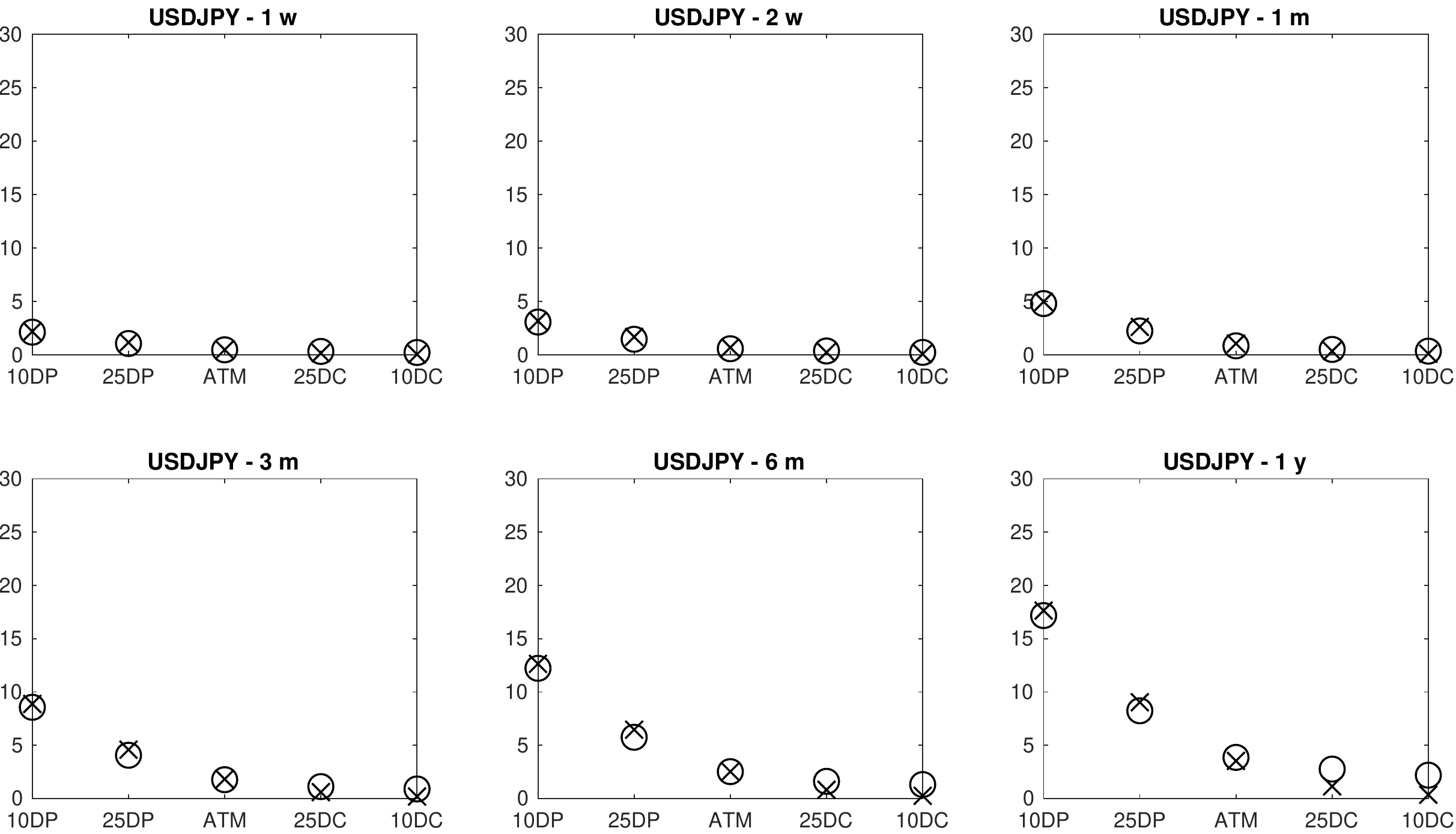}
		\end{subfigure}
		\caption{Calibration results obtained via the deep calibration. Market prices are denoted by crosses, model prices are denoted by circles. Moneyness levels follow the standard Delta quoting convention in the FX option market. DC and DP stand for ``delta call'' and ``delta put'', respectively.}
		\label{fig:cgmydeep}
	\end{figure} 

\section{Conclusions}\label{sec:conclusioncbitclcurrency}
	
We have proposed a stochastic volatility modeling framework for multiple currencies based on CBI-time-changed L\'evy processes (CBITCL processes). The proposed approach combines full analytical tractability with consistency with the symmetric structure and the most relevant risk characteristics of FX markets. In particular, the self-exciting behavior of CBI processes allows capturing jump and volatility clustering effects. We have characterized a class of risk-neutral measures that leave invariant the structure of the model and allow for the derivation of a semi-closed pricing formula for currency options.
Considering a specification driven by tempered $\alpha$-stable CBI processes and CGMY L\'evy processes, we have successfully calibrated the model to  an FX triangle, using standard as well as deep learning techniques. The calibrated values of the parameters support the relevance of self-excitation and clustering phenomena.

Among the possible directions for further research, the modeling framework can be extended by considering stochastic interest rates in the different economies, possibly stochastically correlated with the FX rates. Moreover, we believe that CBITCL processes represent a flexible tool that can be successfully applied to other asset classes where stochastic volatility plays a relevant role.

\bibliographystyle{alpha}
\bibliography{biblio_FX}

\end{document}